\documentclass[11pt]{article}

\usepackage[top=1in, bottom=1in, left=1in, right=1in]{geometry}
\usepackage{authblk}
\usepackage{hyperref}
\usepackage[utf8]{inputenc} 
\usepackage{amsmath}
\usepackage{amsfonts}
\usepackage{amssymb}
\usepackage{siunitx}
\usepackage{graphicx}
\usepackage{subcaption}
\usepackage{float}
\usepackage{amsthm}
\usepackage{thmtools} 
\usepackage{thm-restate}
\usepackage{booktabs}
\usepackage[toc,page]{appendix}
\usepackage[ruled, linesnumbered, noend]{algorithm2e}
\usepackage{setspace}
\usepackage[makeroom]{cancel}
\usepackage{enumitem}
\theoremstyle{plain} 
\newtheorem{thm}{Theorem}[section]
\newtheorem{lem}[thm]{Lemma}

\newtheorem{cor}[thm]{Corollary}

\theoremstyle{definition}

\makeatletter
\newcommand*{\rom}[1]{\expandafter\@slowromancap\romannumeral #1@}
\makeatother
\theoremstyle{remark}

\newtheorem{claim}{Claim}
\newcommand{\E}{\mathbb{E}}
\newcommand{\eps}{\varepsilon}

\newcommand{\R}{\mathbb{R}}

\global\long\def\dt#1#2{\langle #1, #2 \rangle} 

\global\long\def\norm#1{\left\Vert #1\right\Vert } 

\newcommand{\script}[1]{\mathcal{#1}}
\newcommand{\ceil}[1]{\lceil #1 \rceil}
\newcommand{\floor}[1]{\lfloor #1 \rfloor}
\newcommand{\disc}{\mathrm{disc}}
\usepackage{titlesec}
\titlespacing*{\paragraph}{0pt}{1ex plus 1ex minus .2ex}{1em}

 

\usepackage{color}

\title{A Unified Approach to Discrepancy Minimization}
\author{Nikhil Bansal\thanks{University of Michigan, bansaln@umich.edu}, Aditi Laddha\thanks{Georgia Tech, aladdha6@gatech.edu}, Santosh S. Vempala\thanks{Georgia Tech, vempala@gatech.edu}}

\begin{document}

\maketitle

\begin{abstract}
We study a unified approach and algorithm for constructive discrepancy minimization based on a stochastic process. By varying the parameters of the process, one can recover various state-of-the-art results. We demonstrate the flexibility of the method by deriving a discrepancy bound for smoothed instances, which interpolates between known bounds for worst-case and random instances. 

\end{abstract}
\section{Introduction}
Given a universe of elements $U=\{1,\ldots, n\}$ and a collection $\mathcal{S} = \{S_1, \ldots, S_m\}$ of subsets $S_i \subseteq U$, the discrepancy of the set system $\mathcal{S}$ is defined as 
\[
    \mathrm{disc}(\mathcal{S}) = \min_{x: U \rightarrow \{-1,1\}} \max_{i \in [m]} \Big|\sum_{j \in S_i}  x(j) \Big| \,. \]
That is, the discrepancy is the minimum imbalance that must occur in at least one of the sets in $\mathcal{S}$ over all bipartitions of $U$.
More generally for an $m \times n$ matrix $A$, the discrepancy of $A$ is defined as 
$  \mathrm{disc}(A) = \min_{x\in\{-1,1\}^n}\norm{{A}x}_{\infty}$. 
Note that the definition for set systems corresponds to choosing $A$ as the incidence matrix of $\mathcal{S}$, i.e., $A_{ij} = 1$ if $j \in S_i$ and $0$ otherwise.
Discrepancy is a well-studied area with several applications in both mathematics and theoretical computer science (see \cite{beck1995discrepancy, chazelle2001discrepancy, matousek1999geometric}).
\allowdisplaybreaks

\paragraph{Spencer's problem.}
In a celebrated result, Spencer \cite{spencer1985six} showed that the discrepancy of any set system with $m = n$ sets is $O(\sqrt{n})$, and more generally $O(\sqrt{n\log(2m/n)})$ for $m \geq n$.
To show this, he developed a general partial-coloring method (a.k.a. the entropy method), building on a counting argument of Beck \cite{beck1981roth}, that has since been used widely for various other problems. A similar approach was developed independently by Gluskin \cite{gluskin1989extremal}.
Roughly, here the elements are colored in $O(\log n)$ phases. In each phase, an $\Omega(1)$ fraction of the elements get colored while incurring a small discrepancy for each row.

\paragraph{Beck-Fiala and Koml\'{o}s problems.}
Another central question is the Beck-Fiala problem where each element appears in at most $k$ sets in $\mathcal{S}$. Equivalently, every column of the incidence matrix is $k$-sparse. The long-standing Beck-Fiala conjecture  \cite{beck1981integer} states that 
$\mathrm{disc}(\mathcal{S}) = O(\sqrt{k})$.
A further generalization is the Koml\'{o}s problem, also called the vector balancing problem, about the discrepancy of matrices $A$ with column $\ell_2$-norms at most $1$. 
Koml\'{o}s conjectured that $\disc(A) = O(1)$ for any such matrix.
Note that the Koml\'{o}s conjecture implies the Beck-Fiala conjecture.

Banaszczyk showed an $O(\sqrt{\log n})$ bound for the Koml\'{o}s problem  based on a deep geometric result~\cite{banaszczyk1998balancing}.
Here, the full coloring is constructed directly (in a single phase), and this result has also found several applications.
The resulting $O(\sqrt{ k \log n})$ bound for the Beck-Fiala problem is also the best known bound for general $k$.\footnote{For $k =o(\log n)$ an improved bound follows from the $2k-1$ bound by \cite{beck1981integer}.}
 
In contrast, the partial coloring method only gives weaker bounds of $O(\log n)$  and $O(k^{1/2} \log n)$ for these problems -- the $O(\log n)$ loss is incurred due to the $O(\log n)$ phases of partial coloring.

\paragraph{Limitations of Banaszczyk's result.}
Even though Banaszczyk's method gives better bounds for the  Koml\'{o}s problem, it is not necessarily stronger, and
is incomparable to the partial coloring method.
E.g., it is not known how to obtain Spencer's $O(\sqrt{n})$ result (or anything better than the trivial $O(\sqrt{n \log n})$ random-coloring bound) using Banaszczyk's result.
A very interesting question is whether there is a common generalization that unifies both these results and techniques.

\paragraph{Algorithmic approaches.}
Both the partial coloring method and Banaszczyk's result were originally non-algorithmic, and a lot of recent progress has resulted in their algorithmic versions.
Starting with the work of \cite{bansal2010constructive}, several different algorithmic approaches are now known for the partial coloring method \cite{lovett2015constructive, rothvoss2017constructive, harvey2014discrepancy, eldan2014efficient}, based on various elegant ideas from linear algebra, random walks, optimization and convex geometry. 

In further progress,
an algorithmic version of the $O(\sqrt{\log n})$ bound for the Koml\'{o}s problem was obtained by \cite{bansal2019algorithm}, see also \cite{bansal2017algorithmic}, and \cite{bansal2018gram} for the more general algorithmic version of Banaszczyk's result. In related work, Levy et al.~\cite{levy2017deterministic} gave deterministic polynomial time constructive algorithms for the Spencer and Koml\'{o}s settings matching $O(\sqrt{n\log(2m/n)})$ and $O(\sqrt{\log{n}})$ respectively.  

A key underlying idea behind many of these results is to perform a discrete Brownian motion (random walk with small steps) in the $\{-1,1\}^n$ cube, where the update steps are correlated and chosen to lie in some suitable subspace.
However, the way in which these subspaces are chosen for the partial coloring method and the Koml\'{o}s problem are quite different.
We give a high level description of these approaches as this will be crucial later on.

In the partial coloring approach, the walk is performed in a subspace orthogonal to the {\em tight discrepancy constraints}. If the discrepancy for some row $A_i$ reaches its target discrepancy bound, the update $\Delta x$ to the coloring satisfies $A_i \cdot \Delta x=0$. As the walk continues over time, the subspace dimension gets smaller and smaller until the walk is stuck. At this point, the subspace is reset and the {\em next phase} resumes.

On the other hand, the algorithm for the Koml\'{o}s problem 
does not consider the discrepancy constraints at all, 
and chooses 
a different subspace with a certain sub-isotropic property which ensures the discrepancy incurred for a row is roughly proportional to its $\ell_2$ norm, while ensuring that the rows with large $\ell_2$-norm incur zero-discrepancy.
In particular, in contrast to the partial coloring method, all the elements are colored in a {\em single phase}, and the discrepancy constraints are ignored.

\paragraph{The need for a combined approach.} Even though the $O(\sqrt{k \log n})$ bound for the general Beck-Fiala problem is based on Banaszczyk's method, all the important special cases where the conjectured $O(\sqrt{k})$ bound holds are based on the partial coloring method. 
For example, Spencer's problem with $m=O(n)$ sets corresponds to special case of the Beck-Fiala problem with $k = O(n)$. So Spencer's six-deviations result resolves the Beck-Fiala conjecture for this case, which we do not know how to obtain from Banaszczyk's result.

The Beck-Fiala conjecture also holds for the case of random set systems with $m \geq n$. In particular, Potukuchi \cite{potukuchi2019spectral} considers the model where each column has 1's in $k$ randomly chosen rows and shows that the discrepancy is $O(\sqrt{k})$ with high probability. See also \cite{ezra2019beck, bansal2020discrepancy, hoberg2019fourier, altschuler2021discrepancy} for related results.
Potukuchi's result crucially relies on the partial coloring approach, and it is not clear at all how to exploit the properties of random instances in Banaszcyck's approach.

Thus a natural question and a
first step towards resolving the Beck-Fiala and Koml\'{o}s conjecture, and making progress on other discrepancy problems,
is whether there exist more general techniques  
to obtain both Spencer's and Potukuchi's result and the $O(\sqrt{k \log n})$ bound for the Beck-Fiala problem in a unified way.

\subsection{Our results}

We present a new unified framework that recovers all the results mentioned above, and various other state-of-the-art results as special cases.
Our algorithm is based on a derandomization of a stochastic process that is guided by a barrier-based potential function.

Given a matrix $A$, the algorithm starts with the  all-zero coloring $x_0$. Let $x_t \in [-1,1]^n$ be the coloring at time.
The algorithm maintains a  barrier $b_t>0$ over time
and defines the slack of row $i$ at time $t$ as
\begin{equation}\label{slack}  
s_i(t)  = b_t - \underbrace{\sum_{j=1}^n a_i(j) x_t(j)}_{\text{current discrepancy}}    - \lambda  \underbrace{\sum_{j=1}^n a_i(j)^2 (1-x_t(j)^2)}_{\text{remaining variance}}.
\end{equation}
Notice that when all $x_t(j)$ eventually reach $\pm 1$, the {\em remaining variance} term is zero and the slack measures the gap between the discrepancy and the barrier. 
We define the potential  
\begin{equation}\label{potential}
    \Phi(t)  = \sum_i s_i(t)^{-p} 
\end{equation}
for some fixed $p>1$, that penalizes the rows with small slacks and blows up to infinity if some slack approaches zero. If we can ensure that the slacks are always positive and the potential is bounded, then the discrepancy is upper bounded by value of the barrier when the algorithm terminates.

At each time step, the algorithm picks a random direction $v_t$ that is orthogonal to some of the rows with the least slack, and satisfies some additional properties, and updates the coloring by a small amount in the direction $v_t$.
The barrier $b_t$ is also updated. These updates are chosen to ensure that the potential does not increase in expectation, and hence all the slacks stay bounded away from $0$. We give a more detailed overview in Section \ref{sec:framework}.

By changing the parameters $p, \lambda$ depending on the problem at hand, we obtain several results using a unified approach.
\begin{enumerate}
    \item Set coloring~\cite{spencer1985six}.
For any set system on $n$ elements and $m \geq n$ sets,
    $\mathrm{disc}(\script{S}) = O(\sqrt{n\log(2m/n)})$. 
    \item Koml\'{o}s problem~\cite{bansal2017algorithmic}.
For any ${A} \in \mathbb{R}^{m\times n}$ with columns norms $\norm{{A}^{j}}_2 \leq 1$, $\mathrm{disc}(A) = O(\sqrt{\log{n}})$.
\item  Random/Spectral Hypergraphs~\cite{potukuchi2019spectral}.
Let $A \in \{0,1\}^{m\times n}$ be the incidence matrix of a set system with $n$ elements and $m$ sets, where element lies in at most $k$ sets and let $\gamma = \max_{v\perp \mathbf{1}, \norm{v}=1} \norm{Av}$. Then for $m\geq n$,
$\mathrm{disc}(\script{S}) = O(\sqrt{k} + \gamma)$.
\item Gaussian Matrix~\cite{chandrasekaran2014integer}.
For a random matrix $A \in \R^{m\times n}$ with each entry $A_{ij} \sim \script{N}(0, \sigma^2)$ independently, with probability at least $1-(1/m^3)$,
 $\mathrm{disc}(A) = O\left(\sigma\left(\sqrt{n}+\sqrt{\log m}\right) \cdot \sqrt{\log \frac{2m}{n}} \right)$.
\end{enumerate}
More generally, given a matrix $A$, we state the following result based on optimizing the various parameters of the algorithm, depending on the properties of $A$.
This allows our framework to be applied in a black-box manner to a given problem at hand.
 
\begin{restatable}{thm}{mainthm}
\label{thm:g3}
For a matrix $A \in \R^{m\times n}$ with $\norm{A^j}_2\leq L$ and $\vert a_i(j)\vert \leq M$ for all $i\in [m], j\in [n]$, let $h:\R^+ \rightarrow \R^+$ be a non-increasing function such that for every subset $S\subseteq [n]$ and $i \in [m]$, \begin{equation}
 \sum_{j\in S} a_i(j)^2 \leq \vert S\vert\cdot h(\vert S\vert).
    \label{eq:g5}
\end{equation}
Then, for any $p > 1$, there exists a vector $x\in \{-1,1\}^n$ such that
$\norm{Ax}_{\infty} \leq 5b_0 + 2M$,
where
\begin{align}
    b_0 &= \min \left(\sqrt{8(p+1)(48m)^{1/p} \cdot \beta},\;  250 L \sqrt{\log\left(2m\right)} \right)\label{eq:c4}.
\end{align}
where $\beta = \int_{t = 0}^{n-2}h(n-t)(n-t)^{-1/p}dt$.
\end{restatable}

Let us see how Theorem \ref{thm:g3} directly leads to the results stated above.
\paragraph{Set coloring.} As $\|A^j\|_2\leq \sqrt{m}$, we have $L = \sqrt{m}$, and as $\sum_{j\in S} a_i(j)^2\leq \vert S \vert$, we can set $h(t) = 1$ for all $t\in [n]$. 
Consider \eqref{eq:c4} and suppose $p\geq 1.1$ so that $p/(p-1) = O(1)$. Then 
\[ \beta = \int_{t=0}^{n-2} h(n-t)\cdot (n-t)^{-1/p}dt  = O(n^{1-1/p}),\]
and the first bound in \eqref{eq:c4} gives $b_0 = O(pn^{1/2} (m/n)^{1/p})$. Setting $p=\log(2m/n)$ gives Spencer's $O(\sqrt{n \log (2m/n)})$ bound.

Interestingly, the above result gives a new proof of Spencer's six-deviations result based on a direct single-phase coloring. In contrast, all the previously known proofs of this result \cite{bansal2010constructive, lovett2015constructive, rothvoss2017constructive, eldan2014efficient} required 
multiple partial coloring phases.

\paragraph{Koml\'{o}s problem.} Here $L=1$ and the second term in \eqref{eq:c4} directly gives a $O(\sqrt{\log m})$ bound\footnote{It would be interesting to construct an explicit family of examples where the discrepancy obtained by our approach is $\Omega(\sqrt{\log n})$.}. This also implies an $O(\sqrt{\log n})$ bound as at most $n^2$ rows can have $\ell_1$-norm more than $1$, and we can assume that $m \leq n^2$. 

Similarly, bounding $h(t)$ using standard concentration bounds, directly gives the following results for various models of random matrices.

\begin{restatable}[Sub-Gaussian Matrix]{thm}{subg}
\label{thm:subg}
Let $A \in \R^{m\times n}$ with each column drawn independently from a distribution $\script{D}$, where the marginal of each coordinate is sub-Gaussian with mean $0$ and variance $\sigma^2$. Then, for $n \leq m \leq 2^{O(\sqrt{n})}$, $ \mathrm{disc}(A) = O(\sigma\sqrt{n\log(2m/n)})$, with probability at least $1-(1/m^2)$.
\end{restatable}
\begin{restatable}[Random Matrix]{thm}{random}
\label{thm:random}
Let $A \in \R^{m\times n}$, $m\geq n$ such that every column of $A$ is drawn independently from the uniform distribution on $\{x\in \R^m: \norm{x}_2 \leq 1\}$. Then $\mathrm{disc}(A) = O(1)$ with probability at least $1-(1/m^2)$.
\end{restatable}

\subsubsection{Flexibility of the method}
An important advantage of the method is it flexibility, which can be used to obtain several additional results.

\paragraph{Subadditivity.} Given $A, B \in \R^{m \times n}$, can we bound  $\disc(A+B)$ given bounds on $\disc(A)$ and $\disc(B)$? 
Such questions can be directly handled by this framework by considering a weighted combination of two different potential functions -- one for $A$ and another for $B$.

More precisely,
let us define $\mathrm{sdisc}(A)$, the {\em Stochastic Discrepancy} of a matrix $A$, to be the upper bound on discrepancy obtained by the Potential Walk described in Algorithm~\ref{algo:potential_walk}. 
For this notion, we have the following approximate subadditivity for arbitrary matrices.
\begin{restatable}[Subadditivity of Stochastic Discrepancy]{thm}{subadditive}
\label{thm:subadd}
For any two arbitrary matrices $A,B \in \R^{m \times n}$, there exists $x \in \{-1,1\}^n$ such that 
\begin{align*}
 |\dt{a_i}{x}| &\lesssim \mathrm{sdisc}(A) \quad \text{for every row }a_i \text{ of } A, \text{ and}\\
 |\dt{b_i}{x}| &\lesssim \mathrm{sdisc}(B) \quad \text{for every row }b_i \text{ of } B.
\end{align*} \label{thm:sdisc} In particular, this implies that 
 $\mathrm{sdisc}(A+B)\lesssim \mathrm{sdisc}(A)+\mathrm{sdisc}(B)$.
\end{restatable}
Here $a \lesssim b$ means that $a = O(1) b$.
The theorem is algorithmic if $A,B$ are given. It also implies that for any matrix $A$, we have $\mathrm{sdisc}(A) \lesssim \min_B ( \mathrm{sdisc}(B)+\mathrm{sdisc}(A-B))$.

Similar questions have been studied previously in the context of understanding the discrepancy of unions of systems \cite{matouek-detlb, MatousekNikolov15}. For example,
other related quantities such as the $\gamma_2$-norm and the determinant lower bound are also subadditive \cite{matouek-detlb, MatousekNikolov15},
We remark that the additive bound cannot hold for the (actual) discrepancy or even hereditary discrepancy\footnote{A classical example due to Hoffman gives two set systems $A$ and $B$, each with hereditary discrepancy $1$, but their union  has discrepancy $\Omega(\log n/\log \log n)$ \cite{Matousek-geometric}.}, and a logarithmic loss is necessary.
For this reason, the previous additive bounds based on $\gamma_2$-norm and the determinant lower bound lose extra polylogarithmic factors when translated to discrepancy.

A direct application of Theorem \ref{thm:sdisc} is the following.
\begin{restatable}[Semi-Random Koml\'{o}s]{thm}{semirandom}
\label{thm:sr1}
Let $C \in \R^{m\times n}$ be an arbitrary matrix with columns satisfying $\norm{C^j}_2 \leq 1$ for all $j \in [n]$, and $R \in \R^{m \times n}$ be a matrix with entries drawn i.i.d. from $\mathcal{N}(0, \sigma^2)$. Then, for $n \leq m \leq 2^{O(\sqrt{n})}$, with probability at least $1-(1/m^2)$,
\begin{equation*}
    \mathrm{disc}(C+R) = O\left(\sqrt{\log{n}}+ \sigma\sqrt{n\log(2m/n)} \right).
\end{equation*}
\end{restatable}
For $m=O(n)$, the bound above is $O(\sqrt{\log n}+\sigma\sqrt{n})$, which is better than the bound of $O(\sqrt{\log n}(1+\sigma\sqrt{n}))$ obtained by directly applying the best-known bound for the Koml\'{o}s problem to $C+R$.

As another application, consider a matrix $C$ with $n$ columns and two sets of rows, $A$ and $B$, where each row in $A$ has entries in $\{0,1\}$, and the column norm of every column restricted to rows in $B$ is at most $1$. Suppose that $A$ has $O(n)$ rows.
Applying the framework gives a coloring
with $O(\sqrt{n})$ discrepancy for rows in $A$ and $O(\sqrt{\log{n}})$ for rows in $B$.\footnote{This answers a question of Haotian Jiang.}  
Notice that using previous techniques, if we apply the partial coloring method to get $O(\sqrt{n})$ discrepancy for $A$, this would give $O(\log n)$ for rows of $B$. On the other hand, if we apply try to obtain $O(\sqrt{\log n})$ discrepancy for $B$, all the known methods would incur $O(\sqrt{n \log n})$ discrepancy for $A$.

\paragraph{Relaxing the function $h(\cdot)$.} 
Recall that the function $h$ in Theorem \ref{thm:g3}, that controls how the $\ell_2$ norms of rows decrease when restricted to subsets $S$ of columns, and plays an important role in the bounds. 
In many random or pseudo-random instances however, a worst case bound on $h$ can be quite pessimistic. For example, here even though most rows decrease significantly when restricted to $S$, $h$ can remain relatively high due to a few outlier rows.
The following result gives improved bound for such settings where for any subset $S$ of columns, most row sizes restricted to $S$ do not deviate much from their expectation if $S$ is chosen at random.
\begin{restatable}[pseudo-random bounded degree hypergraphs]{thm}{extended}
\label{thm:extended}
Let $A \in \{0,1\}^{m\times n}$ such that $\norm{A^j}_1 \le k$. Suppose there exists $\beta \leq k$ s.t. for any $S\subseteq [n]$ and any $c>0$, the number of rows of $A$ with
\begin{equation}
    \Big\vert \sum_{j\in S} a_i(j) - \norm{a_i}_1 \cdot(\vert S\vert/n)\Big\vert  \geq c\beta\label{eq:s1}
\end{equation}
is at most $c^{-2}\vert S\vert$. Then  $\mathrm{disc}(A) =  O(\sqrt{k} + \beta)$.
\end{restatable}

As discussed in \cite{potukuchi2019spectral}, one can set $\beta \leq \max_{v\perp \mathbf{1}, \norm{v} =1 }\norm{Av}$ in \eqref{eq:s1}, which in particular gives Potukuchi's result \cite{potukuchi2019spectral} for random $k$-regular hypergraphs as $\beta = O(k^{1/2})$ in this case.

Combining with Theorem \ref{thm:subadd}, this extends to the following semi-random setting.
Consider a random $k$-regular hypergraph $A$ with $n$ vertices and $n$ edges. Suppose an adversary can arbitrarily modify $A$ by adding or deleting vertices from edges such that degree of any vertex changes by at most $t$.
How much can this affect the discrepancy of the hypergraph?
\begin{restatable}[Semi-Random Hypergraphs]{thm}{semirandomh}
\label{thm:sr2}
Consider a random $k$-regular hypergraph with incidence matrix $A\in \R^{m \times n}$  with $m\geq n$, and let $C \in \{-1,0,1\}^{m\times n}$ be an arbitrary matrix with at most $t$ non-zero entries per column. Then $\mathrm{disc}(A + C) = O\left(\sqrt{k}+\sqrt{t\log{n}}\right)$ with probability $1-n^{-\Omega(1)}$.
\end{restatable}

\section{The Framework}\label{sec:framework}
Given a matrix $A \in \R^{m\times n}$, we start at some $x_0$ and our goal is to reach an $x_T$ in $\{-1,1\}^n$ with small discrepancy. The basic idea will be to apply a small random update (of size $\delta$) to $x_t$ at step $t$ for $T$ steps, where the update will be chosen with care. 
We use the slack function and the potential function defined in (\ref{slack}) and (\ref{potential}) to implement this approach. The figure below gives a high level description of the algorithm.

\vspace{2mm}

\begin{algorithm}[H]\label{algo:potential_walk}
\SetAlgoLined
\caption{PotentialWalk}
\textbf{Input:} A matrix $A \in R^{m\times n}$, a potential function $\Phi: \R \times \R^n \rightarrow \R^+$.

Let $x_0 = 0, t=0$. Let $T = (n-2)/\delta^2$.

\For {$t \in [T]$}{

Select $v_t$ such that: (i) $\E_{\eps}[\Phi(t+1, x_t + \eps\delta v_t)] \leq \Phi(t, x_t)$, (ii) $x_t \pm\delta v_t  \in [-1,1]^n$, and (iii) $\dt{x_t}{v_t} = 0$, where
$\eps$ is a Rademacher random variable ($\pm 1$ with probability $1/2$).

Let $x_{t+1} = x_t + \eps \delta v_t$. 
 
}
\textbf{Output:} $x_T$
\label{alg:p1}
\end{algorithm}

\subsection{Example: Koml\'{o}s setting}
We first give an overview of the ideas by describing how the framework above works for the Koml\'{o}s setting. Recall that here $A\in \R^{m\times n}$ has columns satisfying $\norm{A^j}_2 \leq 1$. To minimize notation, let us assume here that $m=n$ (this is also the hardest case for the problem).

At time $t$, let $\script{V}_t = \{j \in [n]: \vert x_t(j) \vert < 1- 1/2n \}$ and let $n_t = \vert \script{V}_t\vert$. These are the variables that are ``alive", and not yet ``frozen". To ensure that $x_t \in [-1,1]^n$, the update $v_t$ will only change the variables in $\script{V}_t$.
We also set $\dt{v_t}{x_t} = 0$, which ensures that 
$\norm{x_t}^2 = \delta^2 t$ for any $t \in [0, T]$. 
So $v_t$ satisfies
 \begin{equation}
     v_t(j) = 0 \text{ for all } j\not\in \script{V}_t \text{ and } \dt{v_t}{x_t} = 0.\label{eq:t1}
 \end{equation}
As $\vert x_t(j)\vert \geq (1-1/2n)$ for all $j \notin \script{V}_t$, we have 
\begin{equation*}
   (n-n_t)(1-1/2n)^2\leq \sum_{j\notin \script{V}_t} x_t(j)^2 \leq  \sum_{j \in [n]} x_t(j)^2 = \delta^2 t. 
\end{equation*} 
So the number of alive variables at time $t$ satisfies $n_t \geq n - \frac{\delta^2 t}{(1-(1/(2n)))^2} > n - \delta^2 t - 1$.

{\bf Blocking large rows.} 
To ensure the two-sided bound
$|\sum_j a_i(j) x(j)| < b_0$, we create a new row $-a_i$ for each row $a_i$ at the beginning. Now, as the squared 2-norm of every column of $A$ is at most $2$, at any time $t$, the number of rows with $\sum_{j \in \script{V}_t} a_i(j)^2 > 12$ is at most $\vert \script{V}_t\vert /6 = n_t/6$. 
 Let us call such rows {\em large} (at time $t$). Otherwise, the row is {\em small}.
We additionally constrain $v_t$ so that 
\begin{equation}
    \dt{{a}_i}{v_t} = 0 \text{ for all rows } \{i: \sum_{j\in \script{V}_t} a_i(j)^2 > 12\}. \label{eq:t6}
\end{equation}  
This ensures that a row only starts to incur any discrepancy once it becomes small. So at step $t$, we will define the slacks only for small rows and only such rows will contribute to the potential $\Phi(t)$.  
Let $\script{I}_t$ denote the set of small rows at time $t$. In the slack function \eqref{slack}, we will set $b_t=b_0$ for all $t$ and $\lambda = 2^{-5}b_0$.
So, at the beginning of the algorithm,
when $x_0(j)=0$ for all $j$, we have 
\begin{equation*}
    \Phi(0) = \sum_{i\in \script{I}_0} \frac{1}{(b_0 - \lambda\cdot \sum_{j\in[n]} a_i(j)^2)^p} \leq \frac{|\script{I}_0|}{(b_0 - 12\lambda)^p} \leq n\left(\frac{2}{b_0}\right)^p .
\end{equation*}

At any time $t$, the change in potential $\Phi(t+1)-\Phi(t)$ is due to (i) new rows becoming small and entering $\script{I}_{t+1}$ and (ii) and the change slack of rows in $\script{I}_{t}$.
As each row has discrepancy $0$ until it becomes small, the total contribution of step (i) over the entire algorithm is at most $n (2/b_0)^p$. 

So the main goal will be to show that $\Phi$ does not rise due to step (ii). This will ensure that the potential throughout the algorithm is at most $2n (2/b_0)^p$, which gives the $\sum_j a_i(j) x(j) < b_0$ for all $i$. 

{\bf Bounding the increase in $\Phi$.} 
We now describe the main ideas of the algorithm and computations for the change in $\Phi$ in step (ii).
The desired $O(\sqrt{\log n})$ will then follow directly by optimizing the parameters $b_0$ and $p$ in $\eqref{slack}$.

Let $e_{t,i}$ denote a vector in $\R^n$ with $j$-th entry $a_i(j)^2 x_t(j)$. 
At step $t$, $x_t$ changes as $x_{t+1}-x_t =  \eps \delta \cdot v_t$ and, by a simple calculation, the approximate change in $s_i(t)$ is: 
\begin{align*}
s_i(t+1)-s_i(t) &\simeq \left(2\lambda \dt{e_{t,i}}{v_t} - \dt{{a}_i}{v_t}\right)\eps \delta + \lambda \dt{{a}_i^{(2)}}{v_t^{(2)}} \delta^2
\end{align*}
where $\eps$ is a Rademacher random variable and ${a}^{(2)}$ denotes the vector with $j$-th entry $a(j)^2$. The error terms not included above are all higher powers of $\delta$, and can be ignored for small enough $\delta$ as long as all coefficients are bounded. We formalize this in Section \ref{sec:gen} and Appendix \ref{appendix:ss}.

Then, up to second order terms in $\delta$,
    \[  \Phi(t+1)-\Phi(t) \simeq f(t) \delta^2 + g(t) \eps \delta \]
 where,  
 \begin{align*}
     f(t) &= -p \lambda\sum_{i\in \script{I}} \frac{\dt{{a}_i^{(2)}}{v_t^{(2)}}}{s_i(t)^{p+1}}  + \frac{p(p+1)}{2}\sum_{i\in \script{I}} \frac{\left(2\lambda \dt{e_{t,i}}{v_t} - \dt{{a}_i}{v_t}\right)^2 }{s_i(t)^{p+2}}, \\
     g(t) &= p \sum_{i\in \script{I}} \frac{\left(2\lambda \dt{e_{t,i}}{v_t} - \dt{{a}_i}{v_t}\right)}{s_i(t)^{p+1}}. 
\end{align*}

To bound the expected change in $\Phi$, note that the expectation of the second term  $g(t) \eps \delta $ is zero. 
So it suffices to prove that there is a choice of $v_t$ such that $f(t) \leq 0$. This will ensure the expected change of $\Phi$ is at most $zero$, and 
there will be a choice of $\epsilon$ that ensures $\Phi$ is nonincreasing.

The difficulty in making $f(t)$ at most zero is that the positive part (the second term of $f(t)$) has an extra factor of $s_i(t)$ in the denominator. So if some $s_i(t)$ becomes very small, the positive term could dominate. To ensure this doesn't happen, we choose $v_t$ to be in a subspace that makes this positive term zero for the smallest slack indices.

{\bf Blocking small slacks.} Let $\mathcal{J}_t$ be the subset of $\mathcal{I}$ corresponding to all but the $\floor{n_t/12}$ smallest values of $s_i(t)$ at time $t$. 
Select $v_t$ such that 
\begin{equation}
    \left(2\lambda\dt{e_{t,i}}{v_t} - \dt{{a}_i}{v_t}\right) = 0 \text{ for all } i\in \script{I} \backslash \script{J}_t, \label{eq:t3}
\end{equation} 
Then as $\sum_i s_i(t)^{-p}) \leq \Phi(t)$, and the smallest $n_t/12$ slacks are ``blocked", we have
\[
\max_{j\in \script{J}_t}\frac{1}{s_j(t)} \le \left(\frac{\Phi(t)}{n_t/12}\right)^{1/p},
\]
and so,
\begin{align*}
    f(t) &\leq p\left(\frac{p+1}{2}\sum_{i\in \mathcal{J}_t}  \frac{\left(2\lambda \dt{e_{t,i}}{v_t} - \dt{{a}_i}{v_t}\right)^2}{s_i(t)^{p+1}} \max_{j\in \script{J}_t}s_j(t)^{-1} -\lambda\sum_{i\in \mathcal{I}}  \frac{\dt{{a}_i^{(2)}}{v_t^{(2)}}}{s_i(t)^{p+1}} \right)\\
    &\leq p\left(\frac{p+1}{2}\sum_{i\in \mathcal{J}_t}  \frac{\left(2\lambda \dt{e_{t,i}}{v_t} - \dt{{a}_i}{v_t}\right)^2}{s_i(t)^{p+1}}\left(\frac{12\Phi(t)}{n_t}\right)^{1/p} -\lambda\sum_{i\in \mathcal{I}}  \frac{\dt{{a}_i^{(2)}}{v_t^{(2)}}}{s_i(t)^{p+1}} \right)
\end{align*}
In addition to \eqref{eq:t1} and \eqref{eq:t3}, suppose $v_t$ also satisfies
\begin{equation}
    \sum_{i\in \script{J}_t}  \frac{\dt{2\lambda {e}_{t,i} - {a}_{i}}{v_t}^2}{s_i(t)^{p+1}} \leq 12 \cdot \sum_{i\in \script{J}_t}  \frac{\dt{a_i^{(2)}}{{v_t}^{(2)}}}{s_i(t)^{p+1}}.\label{eq:t4}
\end{equation}
{\bf Choosing the update $v_t$.} Later in Section \ref{sec:gen}, we will see how to find a vector $v_t$ satisfying \eqref{eq:t1}, \eqref{eq:t3}, \eqref{eq:t6}, and \eqref{eq:t4}. Then, 
\[
    f(t)\leq p\sum_{i\in \mathcal{J}_t}  \frac{\dt{{a}_i^{(2)}}{v_t^{(2)}}}{s_i(t)^{p+1}}  \left(6(p+1)\left(\frac{12\Phi(t)}{n_t}\right)^{1/p} - \lambda \right). 
\]
To show that $f(t) \leq 0$, it thus suffices to have 
$ 6(p+1) \left(12\Phi(t)/n_t\right)^{1/p} - \lambda  \leq 0$.

As $\Phi(t)^{\frac{1}{p}} \leq 2(2n)^{1/p}/b_0$ by the inductive hypothesis, and $n_t \geq 1$, it suffices to have
\[\frac{12(p+1)}{b_0} \left(24n\right)^{1/p} - \lambda \leq 0.\]
Choosing $p = \log{n}$ so that $n^{1/p} = O(1)$, and as $\lambda=2^{-5}b_0$, we can pick $b_0 = O(\sqrt{\log{n}})$ to satisfy the above. This gives the desired discrepancy bound. 

\subsection{The General Framework}
\label{sec:gen}
We now describe the algorithm more formally.
Given a matrix $A \in \R^{m\times n}$ with $\norm{A^j}_2 \leq 1$ for all $j\in [n]$, extend $A$ such that for each original row $a_i$ of $A$, there are two rows $a_i$ and $-a_i$ in $A$. Additionally, partition every row $a_i$ into $2$ rows, $a_i^S$ and $a_i^L$, with small and large entries, as follows:
\begin{equation*}
    a_i^S(j) = \begin{cases} 0 & \text{ if } \vert a_{i}(j)\vert > 1/2\lambda \\ a_i(j) & \text{ otherwise}\end{cases}, \quad a_i^L(j) = \begin{cases} a_i(j) & \text{ if } \vert a_{i}(j)\vert > 1/2\lambda\\ 0 & \text{ otherwise,}\end{cases}
\end{equation*}
 where $\lambda$ is a parameter to be determined later.  After this transformation, for any $x \in \R^n$, $\norm{Ax}_{\infty} = \max_{i} \dt{a_i^S + a_i^L}{x}$, and the squared 2-norm of any column of $A$ is at most $2$.

Let $\script{I}$ denote the index set of all rows of ${A}$, and $\script{I}^S$ denote the index set of rows of the first type above.  
 
The step-size of the algorithm is $\delta$ and the algorithm will run for $T = \frac{n-2}{\delta^2}$ steps. Starting with $x_0 = 0$,  let $v_t \in \R^n$ with $\dt{x_t}{v_t} = 0$. For $t \in [T]$,
\begin{equation*}
    x_t = \begin{cases} x_{t-1} + \delta v_{t-1} & \text{ w.p. } 1/2,\\
    x_{t-1} - \delta v_{t-1} & \text{ w.p. } 1/2.
    \end{cases}
\end{equation*}
As $t$ increases, some variables will start approaching $1$ in magnitude. To ensure that $x_t \in [-1,1]^n$, we restrict $v_t$ to be in the space of \emph{alive} variables, defined as  
\begin{equation*}
   \script{V}_t = \{i \in [n]: \vert x_t(i) \vert < 1- 1/(2n) \}. 
\end{equation*} 
For any $t \in [T]$, $\norm{x_t}^2 = \delta^2 t$ as
\begin{equation}
        \norm{x_t}^2 = \norm{x_{t-1} + \delta v_t}^2 = \norm{x_{t-1}}^2 + \delta^2\norm{v_t}^2 = \delta^2(t-1) + \delta^2 = \delta^2t.\label{eq:norm_x}
\end{equation}
Let $n_t = \vert \script{V}_t\vert$ denote the number of alive variables at $t$. By~\eqref{eq:norm_x}, $(n-n_t)(1-\epsilon)^2 \leq \delta^2 t$, which gives \begin{equation*}
    n_t \geq n - \frac{\delta^2 t}{(1-1/(2n))^2} > n - \delta^2 t - 1.
\end{equation*} 

The goal of the rest of this section is to select a $v_t$ such that for all $t\in [T]$, $x_t \in [-1,1]^n$ and $\dt{a_i}{x_t}$ is bounded by some function of $m$ and $n$ for all rows. To help with this goal, we classify the rows according to how many variables are still ``uncolored'' in a row.

Let the set of \emph{s-Alive} rows at time $t$ be defined as: \[
\script{I}_t = \{i \in \script{I}^S: \sum_{j \in \script{V}_t} a_i(j)^2 \leq 20\}.
\]

The choice of $20$ here is arbitrary, and large enough constant works.
We can now define the slack and the potential function.

\noindent \textbf{Slack.} For any $i\in \script{I}$, the slack function is defined as
\[
 s_i(t) = b_t - \dt{a_i}{x_{t}} - \lambda \cdot \sum_{j=1}^{n}a_{i}(j)^2(1-x_{t}(j)^2).\]
We call $b_t$ the barrier,
and for $t\in [T]$, we also move it as
\[
b_t = b_{t-1} + \delta^2 d_{t-1},
\]
for some function $d_{t}$. We set $\lambda = cb_0$ where $c = 1/42$ and $b_0$ is the initial barrier.

\noindent \textbf{Potential function.} The potential function has a parameter $p > 1$ and is defined as\[
    \Phi(t)= \sum_{i\in \mathcal{I}_t} s_i(t)^{-p}.\]
 
\allowdisplaybreaks
We will only consider slacks for {\em alive} rows and ensure that they are always positive. 
Moreover, we will consider only the small {\em s-Alive} rows as the rows in $\script{I}^L$ will be easily handled.
To ensure that $s_i(t)$ does not become too ``small" for any \emph{s-Alive} row, the choice of $v_t$ should not decrease the smallest slacks. This motivates the following definitions.
\begin{itemize}
    \item \emph{Blocked} rows: Let $\mathcal{C}_t$ be the subset of $\mathcal{I}_t$ corresponding to the $\floor{n_t/12}$ smallest values of $s_i(t)$.
    \item Let  $\mathcal{J}_t = \mathcal{I}_t \backslash \mathcal{C}_t$. These are the ``large slack" rows.
\end{itemize}

To prove that all the slacks are positive, we will upper bound the potential throughout
by bounding the change in $\Phi(t)$ at each step.
Note that $\Phi(t)$ will experience jumps whenever a new index gets added to $\script{I}_t$, however the total contribution of jumps is easily shown to be bounded (see Lemma \ref{lem:dphi}) and can essentially be ignored. To bound the one-step change in $\Phi$, we use the second order Taylor expansion of $\Phi(t+1)$ centered at $\Phi(t)$. In Appendix \ref{appendix:ss}, we show that  by choosing $\delta \leq O(1/(n^2m^{6}p^4))$, the overall error due to ignoring the higher powers of $\delta$ is negligible.

\subsection{Algorithm and Analysis}
Recall that $e_{t,i}$ denotes the vector in $\R^n$ with $j$-th entry $a_i(j)^2 x_t(j)$. We can now state the algorithm for selecting $v_t$.

\begin{algorithm}[H]
\label{alg:p2}
\caption{Algorithm for Selecting $v_t$}
\SetAlgoNoLine
\DontPrintSemicolon
\setstretch{1.35}
Initialize $x_0 \leftarrow 0$\;
\For{$t = 1, \ldots, T=\frac{n-2}{\delta^2}$}{
Let $\script{W}_t = \{{w} \in \R^n: {w}(i) = 0, \; \forall i \notin \script{V}_t \}$ \tcp*{\small restrict to alive variables}
Let $\script{U}_t = \{{w} \in \script{W}_t: \dt{{w}}{2\lambda {e}_{t,i} - {a}_{i}} = 0, \forall i \in \mathcal{C}_t \text{ and } \dt{{w}}{x_t} = 0 \}$\;\tcp*{\small restrict to large slack rows}
Let $\script{Y}_t = \{{w} \in \mathcal{W}_t: \dt{{w}}{{a}_{i}} = 0, \forall i \in \script{I} \backslash  \script{I}_t\}$\tcp* {\small restricted to s-Alive rows}
Let $\script{G}_t$ denote the subspace  
\begin{equation}
    \label{eq:gt}
\script{G}_t = \left\{{w} \in \script{W}_t: \sum_{i\in \script{J}_t}  \dt{\left(2\lambda {e}_{t,i} - {a}_{i}\right)}{{w}}^2  s_i(t)^{-p-1} \leq 40 \cdot \sum_{i\in \script{J}_t} \; \dt{a_i^{(2)}}{{w}^{(2)}} s_i(t)^{-p-1} \right\}\end{equation}\\

Consider the subspace $\script{Z}_t = \script{U}_t \cap\script{Y}_t \cap \script{G}_t$ and let $W =\{{w}_1, {w}_2, \ldots, {w}_k \}$ be an orthonormal basis for $\script{Z}_t$. Choose
    \begin{equation}
    \label{eq:vt7}
            v_t = \arg\min_{w \in W} \sum_{i\in \script{J}_t} \dt{2\lambda {e}_{t,i} - {a}_{i}}{w}^2 s_i(t)^{-(p+1)}.
            \end{equation}
}
\end{algorithm}

We now re-state our main theorem. In words, the assumption of the theorem is that there is a non-decreasing function $h(.)$ such that for any row, the squared norm in any subset of coordinates $S$ is proportional to $h(|S|)$ times the size of the subset $S$. Under this condition, we can bound the discrepancy as a function of $h$.

\mainthm*
The case when $h(t)=h$ is often useful, in which case we have following corollary.
\begin{cor}
\label{thm:g2}
For a matrix $A \in \R^{m\times n}$ with $\|A^j\|\leq L$ and $\vert a_i(j)\vert \leq M$ for all $i\in [n], j\in [m]$, let $h$ be such that for every subset $S\subseteq [n]$ and every $i \in [m]$,
\begin{equation}
 \sum_{j\in S} a_i(j)^2\leq \vert S\vert \cdot h.
    \label{eq:g4}
\end{equation}
Then, $\disc(A) \leq 5b_0 + 2M$, where
$b_0 = \min (26\sqrt{h n\log(2m/n)},\; 
    250L\sqrt{\log\left(2m\right)})$.
\end{cor}
\begin{proof}
For a constant $h$, we have
$
  \beta= \int_{0}^{n-2}(n-t)^{-1/p} h dt \leq   n^{1-1/p}h/(1-1/p)$.
Choosing $p = \log(2m/n)$ to optimize the first term in \eqref{eq:c4} gives the result.
\end{proof}

\paragraph{Roadmap of the proof.} The first main lemma below (Lemma ~\ref{lem:g2}) establishes that there is a large feasible subspace from which $v_t$ as defined above can be chosen. Using this we  prove Lemma ~\ref{thm:g1}, which bounds the change in potential. This will allow us to bound the discrepancy of each row and hence prove Theorem~\ref{thm:g3}.

A key fact used for proving Lemma ~\ref{lem:g2} is the following lemma in \cite{bansal2017algorithmic}. We include a proof for the reader's convenience.
\begin{lem}[\cite{bansal2017algorithmic}]
\label{lem:subspace}
Let $G, H \in \R^{m\times n}$ be matrices such that $\vert G_{ij} \vert \leq \alpha \vert H_{ij} \vert$ for all $i\in [m]$ and $j \in [n]$. Let $K = \text{diag}(H^\top H)$. Then for any $\beta \in (0,1]$, there exists a subspace $W \subseteq \mathbb{R}^n$ satisfying
\begin{enumerate}
    \item $\dim(W) \geq (1-\beta)n$, and
    \item $\forall w \in W,\; w^{\top} G^{\top}G w \leq \frac{\alpha^2}{\beta} \cdot w^{\top} K w$.
\end{enumerate}
\end{lem}

\begin{proof}
If $K_{ii} = 0$ for some $i$, then $H_{ji} = G_{ji} = 0$ for all $j \in [n]$. So, for a $w\in W$, $w_i$ can take any value, and removing the $i$-th column of $G$ and $H$ decreases both $n$ and $\dim(W)$ by $1$.  Without loss of generality, assume that $K_{ii} > 0$ for all $i\in [n]$ and let $M = GK^{-\frac{1}{2}}$. For any $w \in \R^n$, let $y =K^{\frac{1}{2}} w$. Then \begin{equation*}
     w^{\top} G^{\top}G w \leq \frac{\alpha^2}{\beta} \cdot w^{\top} K w \Leftrightarrow y^{\top} M^\top M y\leq \frac{\alpha^2}{\beta} \cdot y^\top y.
\end{equation*}Let $Y$ be the subspace of vectors $y$ that satisfy $\beta y^{\top} M^\top M y \leq \alpha^2\cdot y^\top y$. Then $\dim(W) = \dim(Y)$. Thus, $\dim(W)$ is equal to the number of eigenvalues of $M^\top M$ less than $\alpha^2/\beta$. The sum of eigenvalues of $M^\top M$ is equal to $\mathrm{tr}(M^{\top}M)$, which is equal to sum of length squared of columns of $M$. Since $M = GK^{-\frac{1}{2}}$ and $\vert G_{ij} \vert \leq \alpha \vert H_{ij} \vert$, the length of every column of $M$ is at most $\alpha$, and $\mathrm{tr}(M^{\top}M) \leq n\alpha^2$. Therefore, the number of eigenvalues of $M^\top M$ greater than $\alpha^2/\beta$ is at most $\beta n$ and the lemma follows. 
\end{proof}

We now prove Lemma \ref{lem:g2}.
\begin{lem}[Subspace Dimension]
\label{lem:g2}
For all $t\in T$, $\dim(\script{Z}_t) \geq \ceil{2n_t/3}$.
\end{lem}
\begin{proof}
To lower bound the dimension of $\script{Z}_t$ we lower bound the dimensions of $\script{U}_t, \script{Y}_t$ and $\script{G}_t$.

First, we have $\dim(\script{U}_t) \geq n_t - \dim(\script{C}_t)-1 \geq \ceil{11n_t/12}-1$. 
Second, at time $t$,
as the sum of $\ell_2$-norm square of all columns is at most $2n_t$, we have that
$\sum_{i \in \mathcal{I}} \sum_{j\in \mathcal{V}_t} a_i(j)^2 \leq 2n_t$. So the number of rows $a_i$ with $\sum_{j\in \mathcal{V}_t} a_i(j)^2 \geq 20$ is at most $\floor{n_t/10}$ and  $\dim(\script{Y}_t)\geq n_t - \floor{n_t/10} = \ceil{9n_t/10}$. 

We now bound $\dim(\script{G}_t)$ by applying Lemma \ref{lem:subspace}.
Let  $G$ denote the matrix with columns $j$ corresponding to variables in $\script{V}_t$ and rows $i$ restricted to $i \in \mathcal{J}_t$ with $(i,j)$ entry
$(2\lambda {e}_{t,i}(j) - {a}_{i}(j))  s_i(t)^{-(p+1)/2}$.
 
Let $H$
be the matrix with entries ${a}_{i}(j)\cdot s_i(t)^{-(p+1)/2}$ for  $i \in \mathcal{J}_t\}$ and $j \in \mathcal{V}_t$.
As  $\vert a_{ij} \vert \leq 1/(2\lambda)$  for $i \in \script{I}_t$,  we have 
\begin{equation*}
|G_{ij}| =    \vert 2\lambda a_{i}(j)^2x_t(j) - a_j(i) \vert \leq  \vert 2\lambda a_{i}(j)^2x_t(j) \vert + \vert a_j(i) \vert\leq 2\vert a_j(i) \vert = 2 |H_{ij}|.
\end{equation*}

Let $K = \text{diag}(H^\top H)$. Then, using Lemma \ref{lem:subspace} with  $\alpha = 2$ and $\beta = 1/10$, we get that there is a  subspace $\mathcal{G}_t$  with $\dim(\script{G}_t) \geq \ceil{9 n_t/10}$ such that 
\[
    \script{G}_t = \{w \in \script{W}_t: w^\top G^\top G w \leq 40 \cdot w^\top K w\},
\]
which by the definition of $G$ and $H$ is equivalent to that given by \eqref{eq:gt}.

Putting together the bounds on the dimensions of these subspaces gives,
\[
    \dim(\script{Z}_t) \geq \dim(\script{U}_t \cap\script{Y}_t \cap \script{G}_t)\geq \ceil{11n_t/12} - 1 + \ceil{9n_t/10} + \ceil{9n_t/10} - 2n_t \geq \ceil{2n_t/3}.\qedhere\]
\end{proof}

\paragraph{Setting the parameters.}
 To show the two bounds in \eqref{eq:c4}, we will set the parameters $b_t,d_t$ (the change in $b_t$) and $p$ in two ways:
     \begin{equation}
      {\text{ \emph{Case 1:} }} d_t = 4(p+1) \cdot h(n_t) \cdot \max_{i\in \mathcal{J}_t} s_i(t)^{-1}  \text{   for all $t \in [T]$, and $p,b_0$ arbitrary}
   \label{eq:c1}
    \end{equation}
    \begin{equation}
  \text{ \emph{Case 2:} } p = 2\log(2m), \; b_0 =  840(p+1) \cdot \max_{j\in\script{J}_t}s_j(t)^{-1} \text{ and } d_t = 0 \text{ for all }t \in [T].
        \label{eq:c2}
 \end{equation}

\paragraph{Bounding the potential.} The next lemma shows that in both these cases, the potential function remains bounded. 
\begin{lem}[Bounded Potential]
\label{thm:g1} In either of the cases given by \eqref{eq:c1} and\eqref{eq:c2}, we have that 
$\Phi(t) \leq 4m  (2/b_0)^p$,
for all $t=0,\ldots,T$.
\end{lem}
\begin{proof}
We will prove this by induction. Clearly, this holds at $t=0$ as $\Phi(0) \leq  2 m (2/b_0)^p$.
For the inductive step, we will show that for any $j=0,\ldots,T-1$, if $\Phi(j) \leq  4m (2/b_0)^p$ then
\begin{equation}
 \Phi(j+1) \leq \Phi(j) + \frac{1}{Tb_0^p} +  \vert \script{I}_{j+1} \backslash \script{I}_{j} \vert \cdot \left(\frac{2}{b_0}\right)^p.\label{eq:thm_ind}
\end{equation}
Note that $\vert \script{I}_{j+1} \backslash \script{I}_{j} \vert$ is the number of additional rows in $\mathcal{I}^S$ that may become alive at step $j$.
 This gives the result by induction as summing \eqref{eq:thm_ind} over $j=0,\ldots,T-1$  will give
\begin{equation}
    \Phi(t+1) \leq \Phi(0) + \sum_{j=0}^{T-1}\frac{1}{Tb_0^{p}} + \left(\frac{2}{b_0}\right)^p\sum_{j=0}^{T-1}\vert \script{I}_{j+1} \backslash \script{I}_{j} \vert  \leq 2m \cdot \left(\frac{2}{b_0}\right)^p + \frac{1}{b_0^p} \leq 4m \cdot \left(\frac{2}{b_0}\right)^p.\label{eq:phi}
\end{equation}

We now focus on proving \eqref{eq:thm_ind} for $j=t$.

By the induction hypothesis, $\Phi(t)\leq 4m\left(2/b_0\right)^{p}$.
By Lemma \ref{lem:dphi}, one of the signs for $x_{t+1}$ gives
\begin{equation*}
     \E(\Phi(t+1)) -\Phi(t) \leq  f(t)+ \frac{1}{Tnb_0^p}+ \vert \script{I}_{t+1} \backslash \script{I}_{t} \vert \cdot \left(\frac{2}{b_0}\right)^{p}, \text{ where}
\end{equation*}
where
\begin{equation*}
    f(t) = -p \delta^2 \sum_{i\in \script{I}_t}  \frac{d_t + \lambda\dt{{a}_i^{(2)}}{v_t^{(2)}}}{s_i(t)^{p+1}}  + \frac{p(p+1)\delta^2}{2}\sum_{i\in \script{I}_t}  \frac{\left(2\lambda \dt{e_{t,i}}{v_t} - \dt{{a}_i}{v_t}\right)^2 }{s_i(t)^{p+2}}.
\end{equation*}
So to prove \eqref{eq:thm_ind}, it suffices to show that $f(t)\leq 0$. We first consider the case when $b_t,d_t$ and $p$ are given by \eqref{eq:c1}.

As $2\lambda \dt{e_{t,i}}{v_t} - \dt{{a}_i}{v_t} = 0$ for all $i\notin \script{J}_t$, $f(t)$ satisfies 
\begin{align}
 f(t) &\leq -p \delta^2 \sum_{i\in \script{J}_t}  \frac{d_t + \lambda\dt{{a}_i^{(2)}}{v_t^{(2)}}}{s_i(t)^{p+1}}  + \frac{p(p+1)\delta^2}{2}\max_{j\in\script{J}_t}s_j(t)^{-1}\cdot\sum_{i\in \script{J}_t}  \frac{\left(2\lambda \dt{e_{t,i}}{v_t} - \dt{{a}_i}{v_t}\right)^2 }{s_i(t)^{p+1}}. \label{eq:k1}
\end{align}
By a simple averaging argument described in Lemma \ref{lem:g4}, we also have that
\begin{align}
     \sum_{i\in \script{I}_t}  &\frac{\left(2\lambda \dt{e_{t,i}}{v_t} - \dt{{a}_i}{v_t}\right)^2 }{s_i(t)^{p+1}}\leq \sum_{i\in \script{I}_t} \frac{8h(n_t)}{s_i(t)^{p+1}}. \label{eq:g2}
\end{align}
Plugging~\eqref{eq:g2} in~\eqref{eq:k1} gives
\begin{align}
 f(t) &\leq -p \delta^2 \sum_{i\in \script{J}_t}  \frac{d_t}{s_i(t)^{p+1}}  + \frac{p(p+1)\delta^2}{2}\max_{j\in\script{J}_t}s_j(t)^{-1}\cdot\sum_{i\in \script{J}_t}  \frac{8 h(n_t)}{s_i(t)^{p+1}}.
 \label{eq:extra1}
\end{align}
Therefore, if $d_t$ satisfies equation~\eqref{eq:c1}, then $f(t)\leq 0$.

We now consider the case in \eqref{eq:c2}. As $v_t \in \script{G}_t$, we have
\begin{align}
     \sum_{i\in \script{J}_t}  &\frac{\left(2\lambda \dt{e_{t,i}}{v_t} - \dt{{a}_i}{v_t}\right)^2 }{s_i(t)^{p+1}}\leq  40 \cdot \sum_{i\in \script{J}_t}  \frac{\dt{a_i^{(2)}}{{v_t}^{(2)}}}{s_i(t)^{p+1}}.\label{eq:g3}
\end{align}
Next, as $d_t=0$ and $\lambda = b_0/42$, \eqref{eq:k1} and \eqref{eq:g3} give
\[
f(t) 
\label{eq:form2}
\leq  \sum_{i\in \script{J}_t}\frac{p\delta^2\dt{{a}_i^{(2)}}{v_t^{(2)}}}{s_i(t)^{p+1}}\cdot \left(-\frac{b_0}{42}+ 20(p+1) \cdot \max_{j\in\script{J}_t}s_j(t)^{-1} \right). \]
So if $b_0$ satisfies equation~\eqref{eq:c2}, then $f(t)\leq 0$.
\end{proof}
The next lemma gives a bound on the minimum value of slack for any active row, given the bound on potential function.
\begin{lem}
 For any $t \in \{0,\ldots,T\}$, if $\Phi(t) \leq 4m  (2/b_0)^p$, then $ \max_{i\in \mathcal{J}_t} s_i(t)^{-1} \leq \frac{2}{b_0}\left(\frac{48m}{n_t}\right)^{\frac{1}{p}}$.
 \end{lem}
\begin{proof}
By the definition of $\mathcal{J}_t$, for any $i \in \mathcal{J}_t$, there are at least $\floor{n_t/12}+1$ indices $j$ in $\script{I}_t$ such that $s_j(t) \leq s_i(t)$. Therefore, 
\begin{equation}
    \max_{i\in \mathcal{J}_t} \frac{1}{s_i(t)} \leq \left(\frac{12\Phi(t)}{n_t}\right)^{\frac{1}{p}} \leq \frac{2}{b_0}\left(\frac{48m}{n_t}\right)^{\frac{1}{p}},\label{eq:k2}
\end{equation} where the last inequality follows by the assumption, $\Phi(t) \leq 4m  (2/b_0)^p$.
\end{proof}
\begin{lem}
\label{lem:g4}
For any $t\in [T]$, the choice of $v_t$ satisfies
\begin{equation}
  \sum_{i\in \script{J}_t}  \frac{\dt{2\lambda{e}_{t,i} - {a}_{i}}{v_t}^2}{s_i(t)^{p+1}}\leq \sum_{i\in \script{J}_t} \frac{8h(n_t)}{s_i(t)^{p+1}}.\label{eq:vt1}
\end{equation}
\end{lem}
\begin{proof}
 
Using $(a+b)^2 \leq 2 (a^2 + b^2)$, and 
as $ |2\lambda e_{t,i}(j)| = |2 \lambda a_i(j)^2x_t(j)| \leq |a_i(j)|$ as $|a_i(j)|\leq  1/2\lambda$ for any $j$ and $i \in \script{I}^S$, we have that for any $w$,
\begin{equation*}
   \sum_{i\in \script{J}_t}  \frac{\dt{2\lambda{e}_{t,i} - {a}_{i}}{w}^2}{s_i(t)^{p+1}} \leq   \sum_{i\in \script{J}_t}  \frac{2\dt{a_i}{w}^2 + 2\dt{2 \lambda e_{t,i}}{w}^2}{s_i(t)^{p+1}}  \leq 4 \sum_{i\in \script{J}_t}  \frac{\dt{a_i}{w}^2}{s_i(t)^{p+1}}. 
\end{equation*} 
Let $W_t = \{w_1,\ldots,w_k\}$ be an orthonormal basis for $\script{Z}_t$ and $k=\dim(\script{Z}_t)$. As $\script{Z}_t \subseteq \script{V}_t$,
\begin{align*}
  \sum_{i\in\script{J}_t}\frac{\sum_{j=1}^k \dt{a_i}{w_j}^2 }{s_i(t)^{p+1}} \leq \sum_{i\in\script{J}_t}\frac{\sum_{j\in \script{V}_t}a_i(j)^2}{s_i(t)^{p+1}} \leq n_t \sum_{i\in \script{J}_t} \frac{h(n_t)}{s_i(t)^{p+1}}\label{eq:alph}.
\end{align*}
where the second inequality uses that $\sum_{j\in \script{V}_t}a_i(j)^2 \leq n_t \cdot h(n_t)$ by the definition of $h$.

As $k \geq \ceil{n_t/2}$, this gives 
\begin{equation*}
     \frac{1}{k}\sum_{j=1}^k\sum_{i\in \script{J}_t}  \frac{\dt{2\lambda{e}_{t,i} - {a}_{i}}{w_j}^2}{s_i(t)^{p+1}}\leq \frac{n_t}{k} \sum_{i\in \script{J}_t} \frac{4h(n_t)}{s_i(t)^{p+1}} \leq \sum_{i\in \script{J}_t} \frac{8h(n_t)}{s_i(t)^{p+1}} .
\end{equation*}
The result now follows as $v_t$ in \eqref{eq:vt7} minimizes $\sum_{i\in \script{J}_t}  \dt{2\lambda{e}_{t,i} - {a}_{i}}{w_j}^2  s_i(t)^{-p-1}  $ over all $w_j \in W_t$.
\end{proof}

We now prove prove the main theorem.
\begin{proof}[Proof of Theorem \ref{thm:g3}]
Recall that we divide each row $a$ of $A$ as $a = a^S + a^L$. We will bound $\dt{a^L}{x_T}$ and $\dt{a^S}{x_T}$ separately.

Let $t_1$ denote the earliest when the squared norm of $a^L$ (restricted to the alive variables) is at most $20$, and let $n_1$ be number of non-zeros in ${a}^L$ restricted to the set $\script{V}_{t_1}$. 
 As  $\vert a^L(j)\vert \geq 1/(2\lambda)$ for each $j$,  the number of non-zero variables $n_1$ in $a^L$ at time $t_1$ is at most  $80 \lambda^2$, as
 \[   n_1/(4\lambda^2)\leq \sum_{j \in \script{V}_{t_1}} a^L(j)^2 \leq 20.\]
Moreover, as $a^L$ incurs zero discrepancy until $t_1$, the overall discrepancy satisfies
\begin{align}
    \vert \dt{{a}^L}{{x}_{T}} \vert &=  \vert \dt{{a}^L}{{x}_{t_1}}\vert  + \vert \dt{{a}^L}{x_T - {x}_{t_1}}\vert \leq 0+\sqrt{n_1}\cdot (\sum_{j \in \script{V}_{t_1}} a^L(j)^2)^{1/2} \leq 80\lambda\leq 3b_0. \label{eq:al}
\end{align}

Henceforth, we focus on the rows $a^S$. We first show that the slacks are always positive.
Let $\gamma = b_0/4(4m)^{\frac{1}{p}}$. By Lemma \ref{thm:g1}, for all $t \in [T]$, 
$\Phi(t) \leq  4m(2/b_0)^{p}< \gamma^{-p}$.
This implies that $\vert s_i(t) \vert \geq \gamma$ for all $i\in \script{I}^S_t$ and $t \in [T]$. In one step of the algorithm,
\begin{align*}
     \vert s_i(t) - s_i(t-1)\vert&\leq \delta^2 d_{t-1} + \vert \dt{a_i}{x_t} - \dt{a_i}{x_{t-1}}\vert \\
     &\leq \delta^2 d_{t-1} + \vert \delta\dt{a_i}{v_{t-1}} \vert \leq 20 n\delta \leq 2\gamma.
\end{align*}
So, if $s_i(t-1) \geq \gamma$ and $\Phi(t) <\gamma^{-p}$, then $s_i(t) \geq 0$, i.e., the slack $s_i(t)$ cannot go from being greater than $\gamma$ to less than $-\gamma$ in a single step.
So, for every $i\in \script{I}^S$ and $t\in [T]$, $s_i(t)\geq \gamma$ and
$\dt{a_i}{x_{T}} \leq b_T$.
Together with~\eqref{eq:al} this gives, $|\dt{a}{x_T}| \leq |\dt{a^S}{x_T}| + |\dt{a^L}{x_T}| \leq b_T +3b_0$.

Let $x\in \{-1,1\}^n$ be obtained from $x_T$ by the rounding $x(j) = \mathrm{sign}(x_{T}(j))$. 
As $T = (n-2)/\delta^2$, $\norm{x_{T}}^2 = n-2$ with $\vert x_\tau(j) \vert \leq 1$ for all $j\in [n]$. After rounding $x_{T}$ to $x$, we have $\norm{x}^2 = n$. For any row $a$ of $A$,  the discrepancy is bounded by 
\[
     \vert \dt{a}{x} \vert = \vert \dt{a}{x_T} \vert + \vert \dt{a}{x - x_T}\vert\leq \vert \dt{a}{x_T}\vert + M\sum_{j=1}^n \vert x(j) - x_{T}(j)\vert 
     \leq b_T + 3b_0 + 2M.
\]
We now consider the two cases for  $b_0$, $d_t$, $p$. If the second case given by \eqref{eq:c2}, then by \eqref{eq:k2}, $b_0 \leq 1680(p+1) \cdot (48m/n_t)^{1/p}/b_0$.     
     As $n_t \geq 1$ for all $t \in [T]$ and $p = \log(2m)$, we have  $\left(48m/n_t\right)^{1/p}\leq 10e$, and setting $b_0 = 250\sqrt{\log(2m)}$ suffices.
Since $d_t = 0$, $b_T=b_0$ and $\norm{Ax}_\infty \leq 4b_0 + 2M$.

In the first case given by~\eqref{eq:c1}, then by~\eqref{eq:k2}, we have $d_t = 8(p+1)(48m)^\frac{1}{p}\cdot \frac{h(n_t)}{b_0 n_t^{1/p}}$  for all $t \in [T]$. Summing $d_t$ over $t$ gives
\[
    b_T - b_0 = \delta^2\sum_{t=0}^{T-1} d_t =8(p+1)(48m)^\frac{1}{p}\delta^2 \cdot\sum_{t=0}^{T-1} h(n_t)/(b_0n_t^{1/p}).\]
As $n_t > n-\delta^2 t - 1 \geq$ and $h$ is non-increasing,
 $\delta^2 \cdot\sum_{t=0}^{T-1} h(n_t) n_t^{-1/p} 
    \leq \beta$,
    so that $b_T \leq  b_0  + 8(p+1)(48m)^{1/p}\beta/b_0$.
Optimizing  $b_0 = (8(p+1)(48m)^{1/p} \beta)^{1/2}  $ gives that $b_T = 2b_0$ and thus
     $\norm{Ax}_\infty \leq b_T + 3b_0 + 2M \leq 5b_0 + 2M\label{eq:b2}$, giving the desired result.

\end{proof}

\section{Applications}
\subsection{Set Coloring}
We bound the discrepancy of a set system $(U, \mathcal{S})$ with $\vert U\vert = n$, $\vert \script{S}\vert = m$, and $m\geq n$.
As $\|A^j\|_2\leq \sqrt{m}$, we have $L = \sqrt{m}$, and as $\sum_{j\in S} a_i(j)^2\leq \vert S \vert$, we can set $h(t) = 1$ for all $t\in [n]$. 
Consider \eqref{eq:c4} and suppose $p\geq 1.1$ so that $p/(p-1) = O(1)$. Then 
\[ \beta = \int_{t=0}^{n-2} h(n-t)\cdot (n-t)^{-1/p}dt  = O(n^{1-1/p}),\]
and the first bound in \eqref{eq:c4} gives $b_0 = O(pn^{1/2} (m/n)^{1/p})$. Setting $p=\log(2m/n)$ gives Spencer's $O(\sqrt{n \log (2m/n)})$ bound.

\subsection{Vector Balancing}\label{vb}
We now consider the discrepancy a matrix $A \in \R^{m \times n}$ with column $\ell_2$-norms at most $1$.

Here $L=1$ and the second term in \eqref{eq:c4} directly gives a $O(\sqrt{\log m})$ bound. This also implies an $O(\sqrt{\log n})$ bound as at most $n^2$ rows can have $\ell_1$-norm more than $1$, and we can assume that $m \leq n^2$.
In particular,
for a row $a_i$ with $\norm{a_i}_2 < 1/n^{1/2}$, we have $\vert \dt{a_i}{x} \vert\leq \norm{a_i}_1 \leq \sqrt{n\norm{a_i}_2} < 1$ and it can be ignored. The sum of squares of elements in $A$ is at most $n$ the number of rows with $\norm{a_i}_2 > 1/n^{1/2}$ is at most $n^2$.

\subsection{Sub-Gaussian Matrices} \label{gm}
Let $X$ be a random variable with $\E(X) = 0$. $X$ is called Sub-Gaussian with variance $\sigma^2$ if its moment generating function satisfies
$
    \E(e^{sX}) \leq e^{\sigma^2 s^2/2}$ for all $s \in \R$.
For a Sub-Gaussian random variable, $\E(X^2) \leq 4\sigma^2$.
\subg*
\begin{proof}
As $a_{i}(j)$ is a Sub-Gaussian with variance $\sigma^2$, 
$a_{i}(j)^2- \E(a_{i}(j)^2)$ 
is a mean zero and sub-exponential random variable with parameter $16\sigma^2$ \cite{vershynin2018high}.
    
For any $S \subseteq [n]$ with $\vert S\vert = s$,
 
Bernstein's inequality  for sub-exponential random variables \cite{vershynin2018high} (Theorem 2.8.1) gives that, 
\begin{equation}
    \Pr(\sum_{j\in S} a_{i}(j)^2 - \E(a_{i}(j)^2) \geq st) \leq \exp(-\min(s^2t^2/16\sigma^4, st/16\sigma^2)).\label{eq:bern} 
\end{equation} 
Setting $t = 96\sigma^2\left(\log (ne/s) + (\log m)/s \right)$ and as $\E( a_i(j)^2) \leq 4 \sigma^2$, and taking a union bound over all the rows and all possible subsets of $s$ columns, we get that,
    \begin{equation}
        \sum_{j \in S} a_{i}^2(j) \leq 100\sigma^2 \vert S\vert \left(\log (ne/| S|) +\log m)/|S|)\right).\label{eq:reg}
    \end{equation}
     for every $S \subseteq [n]$, $i \in [m]$,  with probability at least $1-1/2m^2$.
     
Similarly, as $a_i(j)$ is  sub-Gaussian with mean $0$ and variance $\sigma^2$, with probability at least $1-1/2m^2$,
    we have 
    $    \vert a_{i}(j) \vert \leq 3\sigma \sqrt{\log(mn)}$ 
    for all $i\in [m], j\in [n]$, and thus
the $\ell_2$-norm of a column is at most $L = 3\sqrt{m}\sigma \sqrt{\log(mn)}$ and $M  = 3 \sigma\sqrt{\log mn}$. By~\eqref{eq:reg}, we can set
\begin{equation*}
   h(t) = 100\sigma^2  \left(\log \left(\frac{ne}{t}\right) +\frac{\log m}{t}\right). 
\end{equation*}
A direct computation gives $\beta = \int_0^{n-2} h(n-t) (n-t)^{-1/p} dt = O(\sigma^2  (n^{1-1/p} +  p\log{m})) .$
Using Theorem \ref{thm:g3} with $p = 2\ceil{\log(2m/n)}$, gives 
$b_0 = O(\sigma (p (m/n)^{1/p} (n+ n^{1/p} p \log m))^{1/2})  = O(\sigma n^{1/2} \log (2m/n))$.

Thus, with high probability  
$ \norm{Ax}_{\infty} \leq (5b_0+2M) = O(\sigma \sqrt{n \log (2m/n)})$.

\end{proof}
\subsection{Random Matrices}
The result above directly implies the following bound for random matrices.
\random*
\begin{proof}Consider a random vector $X$ chosen uniformly at random from the unit ball, $\{x\in \R^m : \norm{x}_2 \leq 1\}$. Then every coordinate of $X$ is sub-Gaussian with variance $\sigma^2 = C/\sqrt{m}$, where $C$ is a constant \cite{vershynin2018high} (Theorem 3.4.6, Ex 3.4.7). The result now follows from Theorem \ref{thm:sr1}.
\end{proof}

\section{Flexibility of the Method}
An advantage of the potential function approach is its flexibility. We describe two illustrative applications.  
In Section \ref{s:stoc-disc} we show how the bounds for matrices $A$ and $B$ obtained using the framework can be used to directly give bounds for $C= A+B$ by combining the potentials for $A$ and $B$ in a natural way.
 
In Section \ref{sec:spectral} we consider how the requirement on the function $h(\cdot)$ in Theorem \ref{thm:g3}
can be relaxed, and use it to bound the discrepancy of sparse hypergraphs (the Beck-Fiala setting) satisfying a certain pseudo-randomness condition.

\subsection{Subadditive Stochastic Discrepancy}
\label{s:stoc-disc}
\subadditive*
\begin{proof}
Let $\Phi_1(t)$, $\Phi_2(t)$ be the potential functions corresponding to $A$ and $B$, respectively.
Let the parameters for Algorithm \ref{alg:p2} on $A$ be $b^1_0, p_1, d^1_t, h_1(\cdot)$ and for $B$ be $b^2_0, p_2, d^2_t, h_2(\cdot)$.

Note that it might not be possible to select an update $v_t$ at time $t$, that ensures that both $\Phi_1(t+1) \leq \Phi_1(t)$ and $\Phi_2(t+1) \leq \Phi_2(t)$ hold, but we can find a $v_t$ for which a weighted sum of $\Phi_1(t)$ and $\Phi_2(t)$ decreases at every step.

Consider the potential function \begin{equation*}
    \Phi(t) = \left(b^1_0/2\right)^{p_1}\Phi_1(t) + (b^2_0/2)^{p_2}\Phi_2(t)\,.
\end{equation*}

We apply the same algorithmic framework. For $t=1, \ldots, T$ , select $v_t$ such that $\E(\Phi(t+1)) \leq \Phi(t)$, and select the sign of $\varepsilon$ for which $\Phi(t+1) \leq \Phi(t)$, and set $x_{t+1} = x_t + \epsilon\delta v_t$.
To this end, it suffices to find a $v_t$ such that $\E(\Phi_1(t+1)) \leq \Phi_1(t)$ and $\E(\Phi_2(t+1)) \leq \Phi_2(t)$. 

 Let $\script{Z}^1_t$ and $\script{Z}^2_t$ be the feasible subspaces at step $t$ for $A$ and $B$ respectively from Algorithm \ref{alg:p2}. 
 We will search for $v_t$ in $\script{Z}_t = \script{Z}^1_t \cap \script{Z}^2_t$. By Lemma \ref{lem:g2}, $\dim(\script{Z}^1_t), \dim(\script{Z}^2_t)\geq  \ceil{2n_t/3}$. Therefore,
\begin{equation*}
    \dim(\script{Z}_t) = \dim(\script{Z}^1_t \cap \script{Z}^2_t)\geq \ceil{2n_t/3} + \ceil{2n_t/3} - n_t\geq n_t/3.
\end{equation*}
 
Using Lemma \ref{lem:g4} on $A$ and $B$, along with Markov's inequality implies that there exists a vector $w \in \script{Z}_t$ such that
\begin{equation}
 \sum_{i\in \script{I}^1_t} \frac{\dt{2c b_0^1{e}_{t,i} - {a}_{i}}{w}^2 }{s_i(t)^{p_1+1}}  \leq \sum_{i\in \script{I}^1_t} \frac{25h_1(n_t)}{s_i(t)^{p_1+1}}\quad\text{ and }\quad   \sum_{i\in \script{I}^2_t} \frac{\dt{2cb_0^2{e}_{t,i} - {a}_{i}}{w}^2 }{s_i(t)^{p_2+1}} \leq  \sum_{i\in \script{I}^2_t} \frac{25h_2(n_t)}{s_i(t)^{p_2+1}}.\label{eq:twovt}
\end{equation}
Comparing~\eqref{eq:twovt} with~\eqref{eq:vt1}, the functions $h_1(\cdot)$ and $h_2(\cdot)$ only increase by a constant factor when compared to running Algorithm  \ref{alg:p2} on $A$ and $B$ independently. So it suffices to multiply $d_t^1$ and $d_t^2$ by $4$ to ensure that by Lemma \ref{thm:g1},
 
\begin{equation}
    \E[\Phi_1(t)] -\Phi_1(t-1) \leq \frac{1}{Tn(b_0^1)^{p_1}} \quad \text{and} \quad \E[\Phi_2(t)] -\Phi_2(t-1) \leq \frac{1}{Tn(b_0^2)^{p_2}}. \label{eq:sa1}
\end{equation}
Plugging~\eqref{eq:sa1} in the definition of $\Phi(t)$, we get $ \E[\Phi(t)] -\Phi(t-1) \leq 2/(Tn)$. So one of the two choices of $x_t$ gives $\Phi(t) -\Phi(t-1) \leq 2/(Tn)$. Summing over $t$,
\begin{align*}
    \Phi(t) &\leq  \Phi(0) + \frac{2}{n}
    \leq \left(\frac{b_0^1}{2}\right)^{p_1}\Phi_1(0) + \left(\frac{b_0^2}{2}\right)^{p_2}\Phi_2(0) + \frac{2}{n}. 
\end{align*}
By Lemma \ref{lem:dphi}, $\Phi_1(0)\leq 2m\cdot (2/b_0^1)^{p_1}$ and $\Phi_2(0)\leq 2m\cdot (2/b_0^2)^{p_2}$, thus $\Phi(t)\leq \Phi(0) +  2/n \leq 5m$.
For a row $i \in \script{J}_t^\ell$ for $\ell \in \{1, 2\}$, we have 
$(\floor{n_t/12}+1) \cdot (b^\ell_0/2)^{p_\ell} \cdot s_i(t)^{-p_\ell} \leq \Phi(t) \leq 5m$, which implies that for any $t$, and $\ell \in \{1,2\}$, 
\begin{equation}
     \max_{i\in \script{J}_t^\ell}\; s_i(t)^{-1} \leq \frac{2}{b_0^\ell}\left(\frac{60m}{n_t}\right)^\frac{1}{p_\ell}. \label{eq:sa2}
\end{equation}
Upon comparing~\eqref{eq:sa2} with~\eqref{eq:k2}, notice that $\max_{k\in \script{J}_t^1} \; s_k(t)^{-1}$ and $\max_{k\in \script{J}_t^2} \; s_k(t)^{-1}$ are only a constant factor larger when compared to running Algorithm \ref{alg:p2} on $A$ and $B$ separately, and hence the discrepancies for both $A$ and $B$ are only a constant factor larger.
\end{proof}

\subsection{Discrepancy of Sparse Pseudo-random Hypergraphs} 
\label{sec:spectral}
In this section, we consider $0/1$ matrices that satisfy a certain regularity property, namely, for most rows, the sum of their entries in any subset of columns is close to the sum of the full row scaled by the fraction of columns in the subset. This property is satisfied, e.g., by the matrices that correspond to sparse random hypergraphs. 
In particular, we show the following.

\extended*

\paragraph{Proof outline.} 
At a high level the proof is similar to that of Theorem~\ref{thm:sdisc}, using a weighted potential function. However, rather than just two potentials, we will have to consider a combination of $O(\log n)$ potentials, and it will take some care to make sure this doesn't create an overhead in the discrepancy. We note that the main algorithm remains: at each step choose a vector in a subspace defined by a set of constraints based on the current vector $x_t$.

Consider the case when $A$ has at most $n$ rows and we run Algorithm \ref{alg:p2} on $A$ with the additional constraints that at time $t$, 
\begin{enumerate}[label=(\alph*)]
    \item we ignore all rows with $\sum_{i\in \script{V}_t} a_i(j)  < 20\beta$ from the potential function, and 
    \item we move orthogonal to all rows for which $|\sum_{i\in \script{V}_t} a_i(j) - \norm{a_i}_1\cdot (n_t/n)| \geq 10n_t$.
\end{enumerate}
In the first case, once the size of rows becomes less than $20\beta$ at some step $t$, we will simply bound the discrepancy gained by this row after $t$ by $20\beta$. 

The second set of rows are the one that do not reduce in size proportional to the progress of the coloring.
Using the assumption in the theorem, i.e., \eqref{eq:s1} with $c = 10$, the number of rows for which point (b) is true is at most $n_t/100$. So, for all but $n_t/100$ rows, 
\begin{equation*}
    20\beta \leq \sum_{j \in \script{V}_t} a_i(j) \leq \norm{a_i}_1 \cdot \frac{n_t}{n} + 10\beta.
\end{equation*}
This gives $\beta \leq (1/10)\norm{a_i}_1 \cdot (n_t/n)$ if row $i$ is active and therefore, for all but $n_t/100$ rows, using the assumption of the theorem,
\begin{equation}
   \sum_{j \in \script{V}_t} a_i(j) \leq 2\norm{a_i}_1  \cdot \frac{n_t}{n}. \label{eq:h1}
\end{equation}
So, $h_i(|S|) = 2\norm{a_i}_1/n$ satisfies the bound \eqref{eq:g5} in Theorem \ref{thm:g3} and we obtain 
\[
    |a_i \cdot x_T| = O(\beta) + \min \left( O(\sqrt{p\cdot \norm{a_i}_1 }), O(\sqrt{n\log(2n)})\right)\]
For $p = 2$, $ |a_i \cdot x_T|= O(\beta+\sqrt{\norm{a_i}_1})$. So, the discrepancy of a row is proportional to the square-root of its initial $\ell_1$-norm. Unfortunately, for rows with large initial norms, this can be as large as $O(\sqrt{n})$.

To fix this issue, let us restrict ourselves to the case when all rows have similar initial $\ell_1$-norm, i.e., for all $i$,
\begin{equation*}
   x\cdot k \leq \norm{a_i}_1 < 2x \cdot k.
\end{equation*}
Since every column of $A$ contains at most $k$ ones, the number of rows with $\ell_1$-norm greater than $x \cdot k$ is at most $(k\cdot n)/ (x \cdot n) = n/x$

By \eqref{eq:h1}, for all but $n_t/100$ rows, $\sum_{j \in \script{V}_t} a_i(j) \leq 4x \cdot k \cdot (n_t/n)$.
Note that a row only gains discrepancy when it satisfies both $\sum_{i \in \script{V}_t} a_i(j) < 20k$ and $|\sum_{i \in \script{V}_t} a_i(j) - \norm{a_i}_1 \cdot (n_t/n) | \leq 10 \beta$. This implies that 
\begin{equation*}
    \norm{a_i}_1\cdot (n_t/n) - 10\beta \leq \sum_{i \in \script{V}_t} a_i(j) \leq 20k. 
\end{equation*}
In other words, $\norm{a_i}_1\cdot (n_t/n) \leq 20 k + 10\beta \leq 30 k$. Under the assumption that $\norm{a_i}_1 \geq x \cdot k$ for all rows, we get $(n_t/n) \leq 30/x$.
So, when $n_t \geq 30n/x$, we can set  $h(n_t) = 0$. In other words, the function 
\[
h(|S|) = \begin{cases}
0 & \mbox{ when } |S| \geq 30n/x\\ 
4x \cdot (k/n) & \mbox{ otherwise}
\end{cases}
\]
satisfies \eqref{eq:g5}. This gives $ \int_{t=0}^{n-2} h(n-t) \cdot (n-t)^{-1/p}dt =  O(x^{1/p} \cdot k \cdot n^{-1/p})$, and by Theorem \ref{thm:g3},
\[
    \disc(A) = \beta +\min \left( O(\sqrt{p\cdot k}), O(\sqrt{n\log(2n)})\right) = O(\beta+ \sqrt{k}) \; \text{ for }p = 2.\]
So if we only consider a set rows with similar initial $\ell_1$-norms (within constant factor of each other) at a time, the discrepancy of such a set is bounded by $O(\beta + \sqrt{k})$. This suggests using Theorem \ref{thm:subadd} to bound the discrepancy of union of this set. However, since the initial $\ell_1$-norms of rows can range anywhere from $1$ to $n$, there can be as many as $\log(n)$ sets and corresponding potential functions. Naively applying Theorem \ref{thm:subadd} will give a $\sqrt{\log(n)}$ factor increase in discrepancy, rather than a constant. 

Before discussing how to fix this issue, we formally describe the partition of rows into classes:
 
 \noindent \textbf{Partitioning rows according to $\ell_1$-norm}: First, extend $A$ such that for each original row $a_i$, there are two rows $a_i$ and $-a_i$ in $A$. Since our goal is to prove discrepancy $O(\sqrt{k})$, we can ignore all rows will $\ell_1$-norm less than $\sqrt{k}$. Then $m\leq n\sqrt{k}$ because the number of rows with $\ell_1$-norm greater than $\sqrt{k}$ is at most $2nk/\sqrt{k} = 2n\sqrt{k}$. Let $N = \ceil{\log_2{n/k}}$ and $\script{\script{Q}} = \{0\} \cup [N]$. Partition the rows of $A$ into based on their initial $\ell_1$-norm into $|\script{Q}|=N+1$ classes: 
\begin{itemize}
    \item $\script{A}_0 = \{i \in \script{I}:  \sqrt{k} \leq \norm{a_i}_1 < 2k \} $. 
    \item For each $i\in [N]$, let $\script{A}_i = \{i \in \script{I}:   2^{i} k \leq \norm{a_i}_1 < 2^{i+1}k \} $.
\end{itemize}
The sum of $\ell_1$-norms of rows in $A$ is at most $2nk$, therefore for any $i$, $2^{i}k \vert \script{A}_i \vert \leq 2nk$ and $\vert \script{A}_i\vert \leq 2^{1-i}n$.

To keep the increase in discrepancy a constant factor rather than $\sqrt{\log(n)}$, we carefully distribute the following two resources among these classes at any step:
\begin{itemize}
    \item The number of rows with small slacks that $v_t$ is orthogonal to from each class. Since the total number of rows $v_t$ can move orthogonal to at time $t$ is at most $n_t$, we need to distribute $n_t$ among the classes. See Lemma \ref{lem:sdim} for more details.
    \item The bound on $ \sum_{i\in \script{I}_t \cap \script{A}_q}  \left(2\lambda \dt{e_{t,i}}{v_t} - \dt{{a}_i}{v_t}\right)^2 s_i(t)^{-p-1}$ in terms of $\sum_{i\in \script{I}_t \cap \script{A}_q} h(n_t)s_i(t)^{-p-1}$ for each class $q$.
\end{itemize}
Rows with larger initial $\ell_1$ norm get more of each resource.

We create $N+1$ potential functions $\{\Phi_i(t)\}_{i=0}^N$, one associated with each row partition. The potential functions use the same $p, b_0$ parameters, and $\lambda = cb_0$ with $c = 1/42$, but have different rate of change of barrier functions $d_q(\cdot)$, based on $q$. We will run Algorithm \ref{alg:p2} on each partition separately but use the same $x_t$ and $v_t$ at each step. In this case, we can select parameters to ensure that each potential function is decreasing in expectation (see Lemma \ref{lem:phi_dec}). However, there might not exist a vector $v_t$ that ensure that moving in $v_t$ direction decreases all the potential functions simultaneously. 

To deal with this, we use a weighted combination of $\Phi_q$ as the potential function:
Let \begin{equation}
    \Phi(t) = \frac{1}{k} \cdot  \Phi_0(t)+\sum_{q \geq 1} 2^{2q} \cdot \Phi_q(t). \label{eq:sum_potential}
\end{equation}
For reasoning behind the form of $\Phi(t)$, see Section \ref{sec:step_t}.
\subsubsection{A suitable subspace}
To identify the constrained subspace for the PotentialWalk (Algorithm~\ref{alg:p1}), we use the following definitions.
The set of \emph{Active} rows is defined as 
\[
\script{I}_t = \{i \in \script{I}: \sum_{j \in \script{V}_t} \vert a_i(j) \vert \leq 12 k\}.    
\]
For each class $q$, let $h_q: \R^+ \rightarrow \R$ be a non-increasing function such that for every subset $S \subseteq n$, at most $n_t/16$ rows $i$ from class $\script{A}_q$ violate the condition
\begin{equation}
     \sum_{j \in S} |a_i(j)| \leq |S| \cdot h_q(|S|) \label{eq:hq}
\end{equation}
While following the general framework from Section \ref{sec:gen}, we make three crucial changes:
\begin{itemize}
    \item Move orthogonal to rows with \emph{large deviation}. At step $t$, the $\ell_1$ norm of row $a_i$ will be close to $(n_t / n)\cdot \norm{a_i}_1$ for most rows.  Let $a_{i,t}$ denote a vector in $\R^n$ with $j$-th entry $\mathbf{1}_{j\in \script{V}_t} a_i(j)$, i.e., $a_{i,t}$ is row $a_i$ restricted to the alive coordinates at time $t$.
Then the set of \emph{large deviation} rows consists of rows that deviate significantly from this expected value 
\begin{equation}
    \script{B}_t = \{i \in \script{I}: \vert \norm{a_{i,t}}_1- \norm{a_i}_1 \cdot(n_t/n) \vert \geq 4\beta \}. \label{eq:badrows}
\end{equation}
For any $t\in [T]$, \eqref{eq:s1} implies that $\dim(\script{B}_t) \leq \floor{n_t/16}$.
\item Ignore \emph{Dead} rows. As soon as the $\ell_1$-norm of some row becomes less than $8\beta$, we drop it from the potential function. The set of \emph{dead} rows at step $t$ is defined as
\begin{equation}
    \script{D}_t = \{i \in \script{I}:  \norm{a_{i,t}}_1 \leq 8\beta \}. \label{eq:deadrows}
\end{equation}
For a dead row, rather than keeping track of its discrepancy using a slack function, we will uniformly bound the the additional discrepancy gained by a row after it becomes dead.
\item \emph{Block} rows based on their initial size. For $q\in \script{Q}$, let $\mathcal{C}_t^q$ be the subset of $\mathcal{A}_q \cap \script{I}_t$ corresponding to the $\floor{2^{i-8}n_t^2/n}$ smallest values of $\{s_i(t): i\in\script{A}_q \cap \script{I}_t\}$, and let $\mathcal{J}_t^q = \mathcal{A}_i \backslash \{\mathcal{C}_t^q \cup \script{D}_t\}$.
\end{itemize}
We are ready to state the algorithm for selecting $v_t$.

\allowdisplaybreaks
\begin{algorithm}[H]
\SetAlgoNoLine
\DontPrintSemicolon
\setstretch{1.35}
\caption{Algorithm for Selecting $v_t$}\label{alg:p3}
Let $h_q(n_t) = 2^{q+2}/n$ and $w_q(t) = 2^{5-\frac{q}{4}} \left(\frac{n}{n_t}\right)^{1/4}$\;
\For{$t=1,\ldots, T$}{
Let $\script{W}_t = \{{w} \in \R^n: {w}(i) = 0, \; \forall i \in \script{V}_t \}$ \tcp*{\small restrict to alive variables}
Let $\script{U}_t = \{{w} \in \script{W}_t: \dt{{w}}{2cb_0 {e}_{t,i} - {a}_{i}} = 0, \forall i \in \mathcal{C}_t \text{ and } \dt{{w}}{x_t} = 0 \}$\; \tcp*{\small restrict to large slack rows}
Let $\script{Y}_t = \{{w} \in \mathcal{W}_t: \dt{{w}}{{a}_{i}} = 0, \forall i \in \script{I} \backslash  \script{I}_t\}$\tcp*{\small move orthogonal to large norm rows}
Let $\script{G}_t =\{w \in \script{W}_t: \dt{a_i}{w} = 0, \; \forall i\in \script{B}_t\}$ \;\tcp*{\small move orthogonal to large deviation rows}
Let $\script{Z}_t = \script{U}_t \cap\script{Y}_t \cap \script{G}_t$ and let $W =\{{w}_1, {w}_2, \ldots, {w}_k \}$ be an orthonormal basis for $\script{Z}_t$\;
Let $v_t\in W$ such that for all $q \in \script{Q}$,
    \begin{equation}
         \sum_{i\in \script{J}^q_t}  \dt{2cb_0 {e}_{t,i} - {a}_{i}}{v_t}^2 s_i(t)^{-p-1} \leq 8  w_q(t) \cdot h_q(n_t)\sum_{i\in\script{J}^q_t} s_i(t)^{-p-1}.\label{eq:vt}
    \end{equation}}
\end{algorithm}

We are now ready for the formal proof. We divide it into several subparts. The first part bounds the number of active classes at time $t$, as a slowly increasing function of $t$. Then we derive the specific weights used in the potential function that combines potential functions for each class of rows (based on initial norm). After that we show that there is a large subspace of vectors which all satisfy the desired goal of not increasing the potential value while satisfying all the constraints about inactive rows and variables. Using this we bound the final discrepancy.

\subsubsection{Number of active classes}\label{sec:step_t}
\begin{lem}\label{lem:snum}
At step $t$, the following two conditions hold: (i) The number of classes $q$ for which $\script{A}_q \cap \{\script{I}_t\backslash\{\script{B}_t \cup \script{D}_t\}\} \neq \emptyset$ is at most $\log(16n/n_t)$ and (ii)
     $h_q(t) = 2^{q+2} k/ n$ satisfies~\eqref{eq:hq} for all $q \in \script{Q}$.
\end{lem}
\begin{proof}
Let $\norm{a_{i,t}}_1 = \sum_{j \in \script{V}_t} |a_i(j)|$, i.e., it is the $\ell_1$-norm of row $i$ restricted to $\script{V}_t$.
At step $t$, if $i \in \script{I} \backslash \{\script{B}_t \cup \script{D}_t\}$, then by~\eqref{eq:badrows} and~\eqref{eq:deadrows},  we have  $ 8\beta \leq \norm{a_{i,t}}_1$ and 
\begin{align}
   (n_t/n) \cdot \norm{a_i}_1 - 4\beta &\leq \norm{a_{i,t}}_1 \leq  (n_t/n) \cdot \norm{a_i}_1 +4\beta.\label{eq:s2}
\end{align}
This gives 
    $4\beta  \leq  (n_t/n) \cdot \norm{a_i}_1$ and 
 $
    \norm{a_{i,t}}_1 \leq (2n_t/n) \cdot \norm{a_i}_1 \label{eq:eq2}$.
     
    Moreover, if $i\in \script{A}_q$ then $\norm{a_i}_1 \leq 2^{q+1}k$ and we get
  $\norm{a_{i,t}}_1\leq  (n_t/n) \cdot 2^{q+2}k.$
Therefore $h_q(t) = 2^{q+2}/k$ satisfies~\eqref{eq:hq}.

Furthermore, if $i \in \script{I}_t$, i.e., $\norm{a_{i,t}}_1 \leq 12 k$, by~\eqref{eq:s2} we have $(n_t/n)\cdot \norm{a_i}_1 -4\beta \leq \norm{a_{i,t}}_1 \leq 12k$. As $\beta < k$,
this gives 
 $(n_t/n)\cdot \norm{a_i}_1 \leq 4 \beta + 12 k \leq 16 k.$
So if $i\in\script{I}_t\backslash\{\script{B}_t \cup \script{D}_t\}$, then
\[   4\beta\cdot (n/n_t)\leq \norm{a_i}_1 \leq 16 k \cdot (n/n_t).\]
Note that this condition is dependent only on the initial $\ell_1$-norm of $a_i$. Since $2^q k\leq\norm{a_i}_1 < 2^{q+1}k$ for any $i\in \script{A}_q$, a necessary condition for $\script{A}_q \cap \{\script{I}_t\backslash\{\script{B}_t \cup \script{D}_t\}\} \neq \emptyset$ is
\begin{equation}
   (2\beta/k)\cdot (n/n_t)   \leq  2^{q} \leq 16 \cdot (n/n_t).\label{eq:s4}
\end{equation}
Therefore $q \leq \log(16n/n_t)$.
\end{proof}
Lemma \ref{lem:snum} implies that at any step $t$, the set of active rows is from the first $\log_2(16n/n_t)$ classes of rows. It also helps us define two important parameters associated with a row class $q$. At step $t$, consider a $q\in \script{Q}$ with $\script{A}_q \cap \{\script{I}_t \backslash \{\script{B}_t \cup \script{D}_t\}\}\neq \emptyset$.
\begin{itemize}
    \item Since $n - \delta^2 t-1 < n_t\leq  16\cdot 2^{-q} n$ By~\eqref{eq:s4}.
    For $q \geq 1$, let
\[   t_q := \max\left\{0,\;n\delta^{-2}\left(1-16\cdot 2^{-q} - 1/n\right)\right\}\]
Similarly, let \[
     t_0 := n\delta^{-2} \left(1-16k^{-1/2} - 1/n\right).\]
Before step $t_q$, for any $i \in\script{A}_q$, $\dt{a_i}{v_t} = 0$. Because $s_i(t)$ is a constant till $t_q$, we set $d_q(t) = 0$ for all $t < t_q$.
    \item  On the other hand, $q$ must satisfy $2^q \leq \frac{16n}{n_t}$. Let
\[
    q_t := \arg \max_{i\geq 0} \left\{2^i \leq 16\cdot(n/n_t)\right\}. 
\]
\end{itemize}

\subsubsection{The weighted potential function}
Now we can justify our choice of the potential function. 
If all the potential functions actually decreased at every step of the algorithm, and we could select a $v_t$ that ensured $\max_{i \in \script{A}_q} (\dt{a_{i}}{v_t})^2 \leq k/n$ for all $q$, then using $h_q(t) = k/n$ for all $q \in \script{Q}$, Theorem \ref{thm:g3} gives us 
\begin{equation*}
    \sum_{t=t_q}^n d_q(t) \simeq O(\sqrt{p(2^{1-q}n)^{1/p}})\cdot \sqrt{\int_{t=t_q}^{n-2}(n-t)^{-1/p}dt\cdot (k/n)} = O(2^{\frac{q}{p} - q(1-\frac{1}{p})}\sqrt{k}) = O(\sqrt{k}),
\end{equation*}
for $p=2$.
However, since the potential functions decrease simultaneously only in expectation, there might not exist a $v_t$ such that each potential function decreases when we move along $v_t$. Instead we take a weighted linear combination of the potential functions $\Phi(t)$~\eqref{eq:sum_potential}, and ensure that $\Phi(t)$ is decreasing at each step $t$. Strictly speaking, $\Phi(t)$ is not decreasing over time but actually increasing as row classes with higher $q$ get added in later steps. When we say $\Phi(t)$ is decreasing, we mean that $\Phi(t+1)$ restricted to rows in $\script{I}_t$ is less than $\Phi(t)$ restricted to rows in $\script{I}_t$, i.e., $\sum_{i\in \script{I}_t}1/s_i(t+1)^{-p}\leq \sum_{i\in \script{I}_t}1/s_i(t)^{-p}$.

What should the weights be? First, we need to normalize $\Phi_q(t)$ by $|\script{A}_q(t)|$. However this is not enough as we still want to use $\Phi(t)$ to bound $1/s_i(t)$ for each active row. However, $\Phi(t)$ can be much larger than the $\Phi_q(t)$. 

If we use the sum of normalized potential functions as the potential, consider some $i \in \script{A}_{q}$. Condition~\eqref{eq:s4} implies that at step $t$, there are at most $\log_2(16n/n_t)$ active classes and therefore $\max_{i\in \script{A}_q} (s_i(t))^{-p}\propto \log_2(16n/n_t) \cdot \Phi_q(0)$. This gives 
\begin{align*}
    \sum_{t=t_q}^n d_q(t) &\simeq O\left(\sqrt{p(2^{1-q}n)^{1/p}} \right)\cdot \sqrt{\int_{t=t_q}^{n-2}\left( \log\frac{n}{n_t}\cdot\frac{1}{(n-t)}\right)^{1/p} \cdot (k/n)} \\
    &= O(q2^{\frac{q}{p} - q(1-\frac{1}{p})}\sqrt{k}) = O(q\sqrt{k}),
\end{align*}
for $p = 2$.
Intuitively, a row with a large initial size may acquire high discrepancy because it gets added to the potential function later, when $\Phi(t)$ contains the potentials corresponding to more row classes $q$, and therefore the value of $\Phi(t)$ is actually higher. This suggests that the potential $\Phi_q(t)$ corresponding to a large $q$ should have a higher weight to balance the effect of a large value of $\Phi(t)$, and hence our choice of $\Phi(t)$:
\begin{equation*}
    \Phi(t) =  \frac{1}{k} \cdot  \Phi_0(t)+\sum_{q = 1}^{q_t} 2^{2q} \cdot \Phi_q(t).  
\end{equation*}

\subsubsection{Bounding the discrepancy}
The next lemma gives a bound on $\dim(\script{Z}_t)$ analogous to \ref{lem:g2}.
\begin{lem}
\label{lem:sdim}
For any $t \in [T]$, it holds that $\dim(\script{Z}_t)\geq \ceil{n_t/2}$.
\end{lem}
\begin{proof}
At time $t$, $\script{I}_t$ only consists of rows from class $\script{A}_q$ with $q\leq q_t$. So,
\begin{align*}
     \dim(\script{C}_t) &\leq \sum_{i=0}^{q_t} \dim(\script{C}_t^i)\leq\sum_{i=0}^{q_t} \frac{2^{i-8}n_t^2}{n} 
    \leq \frac{n_t^2}{n}\cdot 2^{q_t-7}\leq \frac{2^{-7}n_t^2}{n} \cdot \frac{16n}{n_t} \leq \frac{n_t}{8}.
\end{align*}
Since the number of rows in $\script{I}_t$ is at most $\floor{n_t/6}$, we have $\dim(\script{Y}_t) \geq n_t - \floor{n_t/6}$. 

By~\eqref{eq:s1}, $\dim(\script{B}_t) \leq \floor{n_t/16}$ and $\dim(\script{G}_t) \geq n_t - \floor{n_t/16}$.
Putting it together, \begin{equation*}
    \dim(\script{Z}_t) \geq \dim(\script{Y}_t) - \dim(\script{B}_t) - \dim(\script{C}_t) - 1 \geq \ceil{n_t/2}.\qedhere
\end{equation*}
\end{proof}

The next lemma is analogous to Lemma \ref{lem:g4}.
\begin{lem}
\label{lem:salph}
For all $t \in [T]$, there exists $v_t \in \script{Z}_t$ such that $\forall q \in \script{Q}$,
\begin{align}
    \sum_{i\in\script{J}^q_t} \dt{2cb_0 {e}_{t,i} - {a}_{i}}{v_t}^2 s_i(t)^{-p-1}&\leq 8w_q(t) \cdot h_q(n_t) \sum_{i\in\script{J}^q_t} s_i(t)^{-p-1} \; .\label{eq:kvt}
\end{align}
\end{lem}
\begin{proof}
By Lemmas \ref{lem:g4} and \ref{lem:sdim}, for each $q \in \script{Q}$, there exists $v_q \in \script{Z}_t$ such that \begin{equation*}
  \sum_{i\in \script{J}^q_t}  \dt{2cb_0{e}_{t,i} - {a}_{i}}{v_q}^2s_i(t)^{-p-1}\leq \frac{n_t}{\dim(\script{Z}_t)} \cdot\sum_{i\in \script{J}^q_t} 4h_q(n_t)s_i(t)^{-p-1} \leq \sum_{i\in \script{J}^q_t} 8h_q(n_t)s_i(t)^{-p-1} .
\end{equation*}
 However, this does not imply that there exists a $v_t$ that satisfies these bounds for all classes simultaneously. Instead, we use Markov's inequality to assign a weight $w_q(t)$ to each class $q$ at step $t$ such that $\sum_{q = 0}^{q_t} w_q^{-1}(t) < 1$, and therefore there exists a vector $v_t \in \script{Z}_t$ such that
 \begin{equation}
     \sum_{i\in \script{I}_t \cap \script{A}_q}  \left(2\lambda \dt{e_{t,i}}{v_t} - \dt{{a}_i}{v_t}\right)^2 s_i(t)^{-p-1} \leq w_q(t) \cdot \sum_{i\in \script{I}_t \cap \script{A}_q} 8h(n_t)s_i(t)^{-p-1} \label{eq:wts}
 \end{equation} and 
 for each class.  Let 
\begin{equation*}
    \script{\script{Q}}_t = \{q \in \script{\script{Q}}: \script{A}_q \cap \{\script{I}_t\backslash\{\script{B}_t \cup \script{D}_t\}\} \neq \emptyset\}.
\end{equation*}
If some row class $q$ is not in $\script{\script{Q}}_t$, then any row $i \in \script{A}_q$ is either dead or frozen or bad. If it is dead, we drop it from the potential and it does not affect~\eqref{eq:wts}. If it is frozen or bad, $\dt{2cb_0e_{i,t} - a_i}{v_t} = 0$ and the condition is trivially satisfied.
So we only need to consider $q \in\script{\script{Q}}_t$. The weight
$w_q = 2^{5-q/4}\left(n/n_t\right)^{1/4}$ suffices as $\sum_{q = 1}^{q_t}2^{q/4-5}\left(n/n_t\right)^{1/4} \leq 1/2$.
\end{proof}

Note that for any row $i \in \script{A}_q$, at $t \leq t_q$, $\dt{2cb_0e_{i,t} - a_i}{v_t} = 0$. So, we can set $d_t^q = 0$ for rows in class $q$. So, by Lemma \ref{thm:g1} and equation~\eqref{eq:wts}, 
\begin{equation}
      d^q(t) = \begin{cases}
      0 & \text{if } t \leq t_q \\
      4(p+1) \cdot w_q(t) \cdot h_q(n_t) \cdot \max_{i\in \mathcal{J}^q_t} s_i(t)^{-1} & \text{otherwise},
      \end{cases}
   \label{eq:sc1}
\end{equation}
implies that there exists a $v_t \in \script{Z}_t$ such that for all $q \in \script{Q}$,\begin{equation}
    \E[\Phi_q(t)] \leq \Phi_q(t-1) + \frac{1}{Tn b_0^p}. \label{eq:s11}
\end{equation}
The next lemma helps us bound the rate of change of $b_q(t)$, which eventually gives a bound on $b_q(T)$ in Theorem \ref{thm:extended}.
\begin{lem}
For any $t \in \{0, \ldots, T\}$, if $\Phi(t) \leq  8n\left(\frac{2}{b_0}\right)^p (\frac{16n}{n_t})$, then 
\begin{equation}
   \max_{j\in\script{J}_t^q}s_j(t)^{-1} \leq \begin{cases}
    k^{1/p} \cdot \frac{2^{1+15/p}}{b_0}\left(\frac{n}{n_t}\right)^{3/p} & \text{if } q=0 \\
    \frac{2^{1+(15-3q)/p}}{b_0}\left(\frac{n}{n_t}\right)^{3/p} & \text{if } q \geq 1. \label{eq:sss}
    \end{cases}
\end{equation}
\end{lem}
\begin{proof}
For any $q$ and $i \in \mathcal{J}_t^q$, there are at least $\floor{2^{q-8}n_t^2/n}$ indices $j$ in $\script{I}_t \cap \script{A}_q$ such that $s_j(t) \leq s_i(t)$. Therefore, for $q \geq 1$,
\begin{align}
    2^{2q}\cdot \frac{2^{q-8}n_t^2}{n} \cdot s_i(t)^{-p}\leq \Phi(t),\label{eq:s6}
\end{align}
and for $q = 0$,
\begin{align}
     \frac{1}{k}\cdot \frac{2^{-8}n_t^2}{n} \cdot s_i(t)^{-p}\leq \Phi(t).\label{eq:s10}
\end{align}
Plugging $\Phi(t) \leq 8n\left(\frac{2}{b_0}\right)^p (\frac{16n}{n_t})$ in~\eqref{eq:s6} and~\eqref{eq:s10} gives the required bounds.
\end{proof}

\begin{lem}
\label{lem:phi_dec}
For value of $p$ and $d_q$ given by \eqref{eq:wts}, for all $t=0,\ldots,T$, we have 
\begin{equation*}
    \Phi(t) \leq \frac{2^7 n^2}{n_t} \cdot \left(\frac{2}{b_0}\right)^p.
\end{equation*}
\end{lem}
\begin{proof}

Plugging~\eqref{eq:s11} in the definition of $\Phi(t)$,
\begin{equation}
     \E(\Phi(t+1)) - \Phi(t) \leq  \frac{1}{Tb_0^p} + \vert  \{\script{I}_{t} \backslash \script{I}_{t-1}\} \cap \script{A}_0\vert \cdot \frac{1}{k}\cdot \left(\frac{2}{b_0}\right)^p+  2^{2q}\cdot \left(\frac{2}{b_0}\right)^p   \cdot\sum_{q\geq 1}\vert  \{\script{I}_{t} \backslash \script{I}_{t-1}\} \cap \script{A}_q \vert .\label{eq:sphi}
\end{equation}
At every step $t$, the algorithm selects the choice of $x_t$ for which the above inequality is true. Summing $\Phi(s)-\Phi(s-1)$ over $s\in [t]$,
\begin{align}
   \Phi(t)&\leq \Phi(0) + \frac{1}{k}\vert \script{I}^0_{t}  \vert \cdot \left(\frac{2}{b_0}\right)^{p}+ \sum_{q\geq 1} 2^{2q}\vert \script{I}^q_{t} \vert \cdot \left(\frac{2}{b_0}\right)^{p}
\end{align}
For any $q \in \script{Q}$, by Lemma \ref{lem:dphi} we have $\Phi_q(0) + \sum_{t}\vert \script{I}^q_{t+1} \backslash \script{I}^q_{t} \vert \cdot (2/b_0)^{p}  \leq  \vert \script{A}_q \vert \cdot (2/b_0)^{p}$. This gives
\begin{align}
   \Phi(t)&\leq \frac{1}{k}\vert \script{A}_0\vert\cdot \left(\frac{2}{b_0}\right)^p + \sum_{1\leq q \leq q_t}2^{2q} \vert \script{A}_q\vert\cdot \left(\frac{2}{b_0}\right)^p 
\end{align}
Using $\vert \script{A}_0\vert \leq 2n/\sqrt{k}$ and $\vert \script{A}_q\vert \leq 2^{1-q}n$ for $q \geq 1$, we get 
\begin{align}
   \Phi(t)&\leq  2n \left(\frac{2}{b_0}\right)^p\left(\frac{1}{\sqrt{k}}+\sum_{q=1}^{q_t} 2^q\right)   \leq 4n\left(\frac{2}{b_0}\right)^p 2^{q_t+1} \leq 2^7 \left(\frac{2}{b_0}\right)^p \left(\frac{n^2}{n_t}\right), \label{eq:s5}
\end{align}
where the last inequality follows from $2^{q_t} \leq 16(n/n_t)$. 
\end{proof}

\begin{proof}[Proof of Theorem \ref{thm:extended}]
If row $i \in \script{A}_q$ becomes dead after step $t-1$, then
\begin{align*}
    \vert \dt{a_i}{x_T} \vert &\leq \vert \dt{a_i}{x_t} \vert+ \vert \dt{a_i^S}{x_T-x_t} \vert \leq  b_t(q) + 2\sum_{j\in \script{V}_t}\vert a_i(j)\vert \\
    &\leq b_T(q) + 2\sum_{j\in \script{V}_t} a_i(j)^2 \leq b_T(q) + 16\beta.
\end{align*}
Substituting the bound on $\max_{i \in \script{J}_q^t} s_i(t)^{-1}$ from~\eqref{eq:sss}, and using $w_q(t) = 2^{5-q/4}\cdot (n/n_t)^{1/4}$ and $h_q(t) = 2^{q+2}/n$, we get 
\begin{equation*}
    d_q(t) = \begin{cases} 0 & \text{ if } t < t_q\\
    9k\cdot\frac{2^{3q/8+14}}{nb_0}\left(\frac{n}{n-\delta^2 t - 1}\right)^{5/8} & \text{ if }  q \geq 1 \text{ and }  t \geq t_q\\
     9k^{\frac{9}{8}}\cdot\frac{2^{14}}{nb_0}\left(\frac{n}{n-\delta^2 t - 1}\right)^{5/8} & \text{ if } q = 0 \text{ and } t \geq t_0.
    \end{cases}
\end{equation*}
For any $q\geq 1$, summing up $d_q(\cdot)$,
\begin{align*}
    b_q(T) &= b_0 + \delta^2 \sum_{t=t_q}^{T-1} d_q(0) \leq \delta^2 \int_{t = t_q}^T\frac{9k\cdot 2^{3q/4+12+(15-3q)/8}}{nb_0}\left(\frac{n}{n-\delta^2 t-1}\right)^{5/8}dt\\
     &\leq b_0+\int_{t = \delta^2 t_q}^{n-2}  \frac{9k\cdot2^{3q/8+14}}{nb_0}\left(\frac{n}{n- t-1}\right)^{5/8}dt\\
    &\leq  b_0 + \frac{2^{19+3q/8}k}{b_0} \cdot n^{-3/8} \cdot (n-\delta^2t_q)^{3/8} = b_0 +\frac{2^{20}k}{b_0}.
    \end{align*}
    For $b_0 = 2^{10}\sqrt{k}$,  $b_q(T) \leq 2^{11}\sqrt{k}$ for all $q \geq 1$. 
    Similar calculation for $q = 0$ show that $b_0 = 2^{10}\sqrt{k}$ and $b_T(0) = 2^{11}\sqrt{k}$ suffice. 

Let $x\in \{-1,1\}^n$ be obtained from $x_T$ by the rounding $x(j) = \mathrm{sign}(x_{T}(j))$. 
Since $T = (n-2)/\delta^2$, $\norm{x_{T}}^2 = n-2$ with $\vert x_T(j) \vert \leq 1$ for all $j\in [n]$. After rounding $x_{T}$ to $x$, $\norm{x}^2 = n$ and
\begin{align*}
    \vert \dt{a_i}{x} \vert &\leq \vert \dt{a_i}{x_T} \vert+ \vert \dt{a_i}{x-x_T} \vert \leq 2b_T + 16\beta + \sum_{j}\vert x(j)- x_T(j)\vert \\
    &\leq b_T + 16\beta + 2.\qedhere
\end{align*}
\end{proof}

\paragraph{Random and Semi-random Sparse Hypergraphs.}
This gives an alternate proof of the result \cite{potukuchi2019spectral} of Potukuchi that $\disc(\script{H}) = O(\sqrt{k})$ for regular random $k$-regular hypergraph $\script{H}$, on $n$ vertices and $m$ edges with $m\geq n$ and $k=o(m^{1/2})$.
In particular, Potukuchi showed that such hypergraphs satisfy condition \eqref{eq:s1} with high probability.

Consider a random $k$-regular hypergraph $A$ with $n$ vertices and $m$ edges as above, but suppose that an adversary can change the graph so that the number of edges incident to $v$ that are added or deleted is at most $t$. Let $A+C$ denote the incidence matrix of this corrupted hypergraph. How much can this affect the discrepancy of the hypergraph?
\semirandomh*
\begin{proof}
 By the subadditive property of stochastic discrepancy, 
$
    \disc(A+C) \leq O(\sqrt{k}) + O(\sqrt{t\log{n}}).$
However, this bound is not algorithmic because it requires running the algorithm separately on $A$ and $A_c - A$.
\end{proof}

\paragraph{Acknowledgments.}
We are grateful to Yin Tat Lee and Mohit Singh for helpful discussions. The latter two authors were supported in part by NSF awards CCF-2007443 and CCF-2134105.

\bibliographystyle{plain}
\bibliography{references}
\begin{appendix}

\section{Appendix: Bounding the step size}\label{appendix:ss}

\begin{lem}\label{lem:dphi}
For $A \in \R^{m \times n}$,
\begin{itemize}
    \item $\Phi(0) + \sum_{t}\vert \script{I}_{t+1} \backslash \script{I}_{t} \vert \cdot \left(\frac{2}{b_0}\right)^{p}  \leq 2m \cdot \left(\frac{2}{b_0}\right)^{p}$.
    \item 
For all $t \in \{0, 1, \ldots, T-1\}$, if $\Phi(t) \leq 2^7 m^2 \left(\frac{2}{b_0}\right)^{p}$ and $d_t = O(pn \cdot \max_{i \in \script{J}_t} s_i(t)^{-1})$, then for step size $\delta^2 \leq (Cn^2m^6p^4)^{-1}$, 
\begin{align*}
    \E(\Phi(t+1)) -\Phi(t) &\leq  f(t)+ \frac{1}{Tnb_0^p}+ \vert \script{I}_{t+1} \backslash \script{I}_{t} \vert \cdot \left(\frac{2}{b_0}\right)^{p}\\
      \text{where }\qquad f(t) &= -p \delta^2 \sum_{i\in \script{I}_t}  \frac{d_t + cb_0\dt{{a}_i^{(2)}}{v_t^{(2)}}}{s_i(t)^{p+1}}  + \frac{p(p+1)\delta^2}{2}\sum_{i\in \script{I}_t}  \frac{ \dt{2cb_0e_{t,i} - a_i}{v_t}^2 }{s_i(t)^{p+2}}.
\end{align*}
\end{itemize}
\end{lem}
\begin{proof}
We note that the purpose of this lemma is to allow the proof to ignore higher order terms by making the step size inverse polynomially small, and thereby obtain a (deterministic) polytime algorithm. As our focus is on establishing polynomiality, we have not optimized the bounds. 

For any $i \in \script{I}_{t+1}\backslash \script{I}_{t}$, we have $\dt{a_i}{x_{t+1}} = 0$ and 
\[
\sum_{j}a_{i}(j)^2(1-x_{t+1}(j)^2) \leq \sum_{j\in\mathcal{V}_{t+1}} a_i(j)^2 + n (1-(1-\frac{1}{2n})^2) < 21.
\]
Therefore, for any $i \in \script{I}_{t+1}\backslash \script{I}_{t}$, using the fact that the coefficient of the above energy term is $cb_0 = b_0/42$,  
\begin{align*}
    \frac{1}{s_{i}(t+1)} \leq \frac{1}{(b_{t+1} -  21 cb_0)}\leq \frac{2}{b_0} .
\end{align*}
Therefore 
\[
\Phi(0) = \sum_{i \in \script{I}_0} s_i(0)^{-p} \leq |\script{I}_{0}| \cdot (\frac{2}{b_0})^p.
\]
Since $|\script{I}_{0}| + \sum_{t} |\script{I}_{t+1}\backslash \script{I}_{t}| \leq 2m$, we have 
\begin{equation*}
    \Phi(0) + \sum_{t}\vert \script{I}_{t+1} \backslash \script{I}_{t} \vert \cdot \left(\frac{2}{b_0}\right)^{p}  \leq 2m \cdot \left(\frac{2}{b_0}\right)^{p}.
\end{equation*}
This concludes the proof of the first part.

For the second part, we will use a second-order Taylor approximation and choose $\delta$ small enough so that the higher order terms are negligible. 

Let $Z_t(b, x) := \sum_{i \in \script{I}_t} \left(b - \dt{a_i}{x} - \lambda\cdot \sum_{j=1}^n a_i(j)^2(1-x(j)^2)\right)^{-p} = \sum_{i \in \script{I}_t} s_i(t)^{-p}$, the potential function restricted to the active rows in time step $t$. 
Then,
\begin{align*}
    \Phi(t+1) - \Phi(t)  &= \sum_{i\in \script{I}_t}s_i(t+1)^{-p} -s_i(t+1)^{-p}  + \sum_{i\in \script{I}_{t+1} \backslash \script{I}_{t} }s_i(t+1)^{-p}\\
    &\leq Z_t(b_{t+1}, x_{t+1}) - Z_t(b_t, x_t) + \vert \script{I}_{t+1} \backslash \script{I}_{t} \vert \cdot \left(\frac{2}{b_0}\right)^{p}. 
\end{align*}
Hence,
\begin{align}
    \E(\Phi(t+1)) - \Phi(t)  &\leq \E(Z_t(b_{t+1}, x_{t+1})) - Z_t(b_t, x_t) + \vert \script{I}_{t+1} \backslash \script{I}_{t} \vert \cdot \left(\frac{2}{b_0}\right)^{p}\label{eq:apf}
\end{align}
Using Taylor's theorem,
\begin{align*}
   Z_t(b_{t+1}, x_{t+1}) - Z_t(b_t, x_t) &= \delta \cdot\nabla_xZ_t(b_t, x_t)^{\top}v_t + \delta^2\cdot\nabla_b Z_t(b_t, x_t)d_t+ \frac{\delta^2}{2}\cdot v_t^{\top} \nabla_x^2Z_t(b_t, x_t) v_t \\ &+ \frac{\delta^4}{2}\cdot \nabla^2_bZ_t(b_t,x_t)d_t^2 + \frac{1}{6}\cdot\nabla^3Z_t(b', x')[w,w,w],
\end{align*}
for some $b' \in [b_t, b_t + \delta^2 d_t]$ and $x' \in [x_t, x_t+\delta v_t]$, and $w$ is the tuple $(\delta^2 d_t, \delta v_t)$. Taking expectation,
\begin{align}
   \E(Z_t(b_{t+1}, x_{t+1})) - Z_t(b_t, x_t) &=  \delta^2\cdot\nabla_b Z_t(b_t, x_t)d_t+ \frac{\delta^2}{2}\cdot v_t^{\top} \nabla_x^2Z_t(b_t, x_t) v_t \notag \\ 
   &+ \frac{\delta^4}{2}\cdot \nabla^2_bZ_t(b_t,x_t)d_t^2 + \E(\frac{1}{6}\cdot\nabla^3Z_t(b', x')[w,w,w]).\label{eq:ap3}
\end{align}

For any $t \in [T]$,
\begin{align}
    \nabla_b Z_t(b_t,x_t) &= -p\sum_{i \in \script{I}_t} \frac{1}{s_i(t)^{p+1}}, \label{eq:ap4}\quad \text{and}\\
     \nabla^2_x Z_t(b_t,x_t) &= p(p+1)\sum_{i \in \script{I}_t} \frac{(2cb_0 e_{t,i}-a_i )(2cb_0 e_{t,i}-a_i)^{\top}}{s_i(t)^{p+2}} -pcb_0\sum_{i \in \script{I}_t}\frac{\text{diag}(a_i^{(2)})}{s_i(t)^{p+1}}.\label{eq:ap5}
\end{align}
We will show the following claim.

\begin{claim}
For any $t$ and any $b',x'$ as defined above, 
\[\E(\frac{1}{6}\cdot \nabla^3Z_t(b', x')[w,w,w]) + \frac{\delta^4}{2}\cdot \nabla^2_bZ_t(b_t,x_t)d_t^2 \leq \frac{1}{Tnb_0^p}.\]
\end{claim} 
Combining this claim with~\eqref{eq:ap3},~\eqref{eq:ap4}, and~\eqref{eq:ap5}, we get 
\begin{align*}
     \E(Z_t(b_{t+1}, x_{t+1})) - Z_t(b_t, x_t) &\leq 
     -p \delta^2 \sum_{i\in \script{I}_t}  \frac{d_t + cb_0\dt{{a}_i^{(2)}}{v_t^{(2)}}}{s_i(t)^{p+1}}  \\&+ \frac{p(p+1)\delta^2}{2}\sum_{i\in \script{I}_t} \frac{ \dt{2cb_0e_{t,i} - a_i}{v_t}^2 }{s_i(t)^{p+2}}
     + \frac{1}{Tb_0^p}\\
     &= f(t)  + \frac{1}{Tnb_0^p}.
\end{align*}
Substituting this bound in~\eqref{eq:apf} proves the lemma.

{\bf Proof of Claim 1.}
 
As $\Phi(t) \leq 2^7 m^2 \cdot \left(2/b_0\right)^p$, for any $i\in \script{I}_t$, \begin{equation}
    s_i(b_t, x_t) = s_i(t) \geq b_0 (2^{p+7}m^2)^{-1/p}. \label{eq:ap1}
\end{equation}
 
By~\eqref{eq:ap1}, \begin{equation}
    d_t = O(pn \cdot \max_{i \in \script{J}_t} s_i(t)^{-1}) = O\left(pn \cdot (2^{p+7} m^2)^{1/p} b_0^{-1}\right).\label{eq:ap2}
\end{equation}
By~\eqref{eq:ap1} and~\eqref{eq:ap2}, and as
the second derivative of $Z_t$ with respect to $b_t$ is \[
     \nabla^2_bZ_t(b_t, x_t) =  p(p+1)\sum_{i\in\script{I}_t} s_i(t)^{-p-2},\] we obtain 
\[\delta^4\nabla^2_bZ_t(b_t, x_t)d_t^2 = O(2^p n^2m^{3+\frac{3}{p}}p^4\delta^4 (m/n)^{\frac{2}{p}}b_0^{-p-4}).\]
For each of the choices $p = 2\ceil{\log(2m)}$ or $p = 2\ceil{\log(2m/n)}$ or $p=8$, since $\delta^2 = 1/Cn^2m^6p^4$ and $T=(n-2)/\delta^2$, we have
\begin{equation}
    \delta^4\nabla^2_b Z_t(t)d_t^2 \leq \frac{\delta^2}{2n(n-2)b_0^p} \leq \frac{1}{2Tnb_0^p}. \label{eq:b1}
\end{equation}
$\E(\nabla^3 Z_t(b', x'))$ in direction $w$ is given by
\begin{align}
&\E(\nabla^3 Z_t(b', x')[w,w,w]) = -p(p+1)(p+2)\sum_{i\in\script{I}_t} \frac{\delta^6 d_t^3}{s_i(b', x')^{p+3}} \notag \\
&-   3p(p+1)(p+2)\delta^4\sum_{i\in\script{I}_t} 
\frac{d_t\left(2cb_0 \dt{(a_i^2 x')}{v_t} - 
\dt{{a}_i}{v_t}\right)^2}{s_i(b', x')^{p+3}} +3p(p+1)cb_0\delta^4\sum_{i\in\script{I}_t} \frac{d_t\dt{{a}_i^{(2)}}{v_t^{(2)}} }{s_i(b', x')^{p+2}}\notag\\
&\leq 3p(p+1)cb_0\delta^4\sum_{i\in\script{I}_t} \frac{d_t}{s_i(b', x')^{p+2}},\label{eq:a3}
\end{align}
where we use that $d_t,s_i \geq 0$.

To bound the difference between $ s_i(b', x') - s_i(b_t, x_t)$, consider the difference between $b'$ and $b$, 
\begin{align}
 \vert b' - b_t \vert &\leq \delta^2 d_t \leq \frac{b_0}{16(2^7m^2)^{1/p}},\label{eq:a1}
\end{align}
and the difference between $\dt{a_i}{x'}+cb_0\cdot \sum_{j=1}^n a_i(j)^2(1-x'(j)^2)$ and $\dt{a_i}{x_t}+cb_0\cdot \sum_{j=1}^n a_i(j)^2(1-x_t(j)^2)$,
 \begin{align}
 \vert \dt{a_i}{x'} &+\sum_{j=1}^n a_i(j)^2(1-x'(j)^2) - \dt{a_i}{x_t} -\sum_{j=1}^n a_i(j)^2(1-x_t(j)^2)\vert  \notag\\&\leq \delta \vert \dt{a_i}{v_t}\vert  + cb_0 \sum_{j} a_i(j)^2\vert x_t(j)^2 - (x_t(j)+\delta\lambda_2 v_t(j))^2)\vert \notag\\
 &\leq    \delta(1+4cb_0\sqrt{n})\leq \frac{b_0}{16(2^7m^2)^{1/p}} \label{eq:a2}. 
\end{align}
By~\eqref{eq:a1} and~\eqref{eq:a2}, \begin{align}
    s_i(b', x') &= s_i(b_t, x_t) + y - b_t + \dt{a_i}{x'} +\sum_{j=1}^n a_i(j)^2(1-x'(j)^2) \dt{a_i}{x_t} -\sum_{j=1}^n a_i(j)^2(1-x_t(j)^2)\notag \\
    &\geq s_i(b_t, x_t)-\frac{b_0}{16(2^7m^2)^{1/p}} -\frac{b_0}{16(2^7m^2)^{1/p}}\geq \frac{3b_0}{8(2^7m^2)^{1/p}}\label{eq:a4}
\end{align}
By~\eqref{eq:a3} and~\eqref{eq:a4}, $\E(\nabla^3 Z_t(b', x')[w,w,w]) = O(nm^{3+\frac{3}{p}}p^3\delta^4 (8/3)^p b_0^{-p-1})$.

Again, since $p = 2\ceil{\log(2m)}$ or $p = 2\ceil{\log(2m/n)}$ or $p=8$, for $\delta^2 = 1/Cn^2m^6p^4$, 
\begin{equation}
    \E(\nabla^3 Z_t(b', x')[w,w,w])\leq \frac{\delta^2}{2n(n-2)b_0^p} = \frac{1}{2Tnb_0^p}.  \label{eq:b0}
\end{equation}
Now the claim follows from the bounds~\eqref{eq:b1} and~\eqref{eq:b0}.
\end{proof}
\end{appendix}
\end{document}